\documentclass[sigplan,10pt]{acmart}

\usepackage{lineno}
\usepackage{subcaption}
\usepackage{enumitem}
\usepackage{makecell}
\usepackage{multirow}
\usepackage{booktabs}
\usepackage{listings}
\usepackage{tabularx}
\usepackage{amsthm}

\newtheorem{theorem}{Theorem}

\newif\ifdraft
\newcommand{\sk}[1]{\ifdraft{\noindent{\textcolor{blue}{\bf \fbox{SL} {\it#1}}}}\fi}

\newcommand{\kr}[1]{\ifdraft{\noindent{\textcolor{orange}{\bf \fbox{KR} {\it#1}}}}\fi}

\renewcommand\footnotetextcopyrightpermission[1]{} 

\setcopyright{none}

\acmDOI{}

\acmISBN{}

\acmPrice{}

\begin{document}

\title{Libra: Accelerating Socket I/O via Programmable Selective Data Copying}

\author{Kairui Zhou}
\affiliation{%
  \institution{Shanghai Jiao Tong University}
  \country{China}
}
\email{kairuizhou@sjtu.edu.cn}

\author{Shengkai Lin}
\affiliation{%
  \institution{Shanghai Jiao Tong University}
  \country{China}
}
\email{jefflin@sjtu.edu.cn}

\author{Wei Zhang}
\affiliation{%
  \institution{University of Connecticut}
  \country{USA}
}
\email{wei.13.zhang@uconn.edu}

\author{Shizhen Zhao}
\affiliation{%
  \institution{Shanghai Jiao Tong University}
  \country{China}
}
\email{shizhenzhao@sjtu.edu.cn}

\renewcommand{\shortauthors}{K. Zhou et al.}

\begin{abstract}




Layer-7 (L7) proxies are critical to modern cloud-native systems, yet their performance is increasingly bottlenecked by copying entire payloads across the kernel-user boundary. Existing approaches reduce this overhead but typically sacrifice compatibility with unmodified POSIX applications, introduce new APIs, or require specialized environments. We show that, under conventional OS abstractions, fully eliminating kernel-user copies while preserving standard socket semantics for unmodified proxies is fundamentally impossible. This leads to a practical insight: in common L7 workloads, proxies inspect only small metadata (e.g., HTTP headers) for routing, while forwarding the bulk payload unchanged. Based on this insight, we present Libra, an OS-level selective-copy framework that copies only metadata to the user space and retains the bulk payload in the kernel for forwarding, reducing data movement without breaking compatibility. Libra uses eBPF to identify protocol-specific metadata boundaries and coordinate selective copy and payload reuse across receive and transmit paths, all without modifying the socket API. Implemented in Linux and evaluated with unmodified Nginx and HAProxy, Libra improves plaintext throughput by up to 4.2$\times$ and reduces P99 tail latency by over 90\%. With hardware-offloaded kTLS, it boosts encrypted throughput by 2.0$\times$ and cuts tail latency by 65\%.

\end{abstract}

\maketitle
\pagestyle{plain}

\section{Introduction}
\label{sec:intro}

Layer-7 (L7) proxies are foundational components of modern cloud environments, serving as intermediaries that inspect application-layer data—such as HTTP headers, URLs, or methods—to make routing, load-balancing, and access control decisions~\cite{aws-alb}. Widely deployed as API gateways, ingress controllers, and service mesh sidecars, they act as middlemen that form the critical communication backbone of cloud-native platforms. Today, these tasks are powered by highly optimized production-grade implementations such as Nginx~\cite{nginx}, Envoy~\cite{envoy}, HAProxy~\cite{haproxy}, and Traefik~\cite{traefik}.

In cloud computing environments, Layer-7 (L7) proxies handle a massive number of concurrent requests~\cite{hermes}. The entire payload of each request must be copied between kernel and user space, incurring extremely high CPU overhead. As our empirical profiling of production-grade proxies shows (\S\ref{sec:copies_in_modern_l7_proxies}), memory copying across the kernel-user boundary for both ingress and egress traffic combined can consume up to 55.6\% of total CPU cycles \sk{For ingress or egress?}, making it the dominant performance bottleneck~\cite{copier, zio, mallacc, kanev2015profiling, gonzalez2023profiling}.

To alleviate the severe performance bottlenecks of L7 proxies, several targeted optimizations have been proposed~\cite{tcp_splice, prism, qdsr, miresga}. Broadening this scope, general-purpose optimizations have also been developed across various layers of the system stack to reduce data movement overheads~\cite{arrakis, IX, zerocopy-recv, zerocopy-send, freebsd, solaris, socksdirect, copier}. 

To evaluate their limitations, we consider two key properties: \textit{zero-copy}, where once data has been placed in host memory, neither \texttt{recv} nor \texttt{send} performs any memory-to-memory data movement, and \textit{POSIX compatibility}, i.e., support for unmodified applications through standard POSIX APIs. A precise definition of both terms is provided in ~\S\ref{sec:proof}.

We summarize these existing optimizations based on these two key criteria. As shown in Table~\ref{tab:related_work}, prior approaches fall into several categories, such as new syscalls, dedicated proxies, pre-registered memory or memory remapping. 

Although these optimizations reduce data copying, none of them simultaneously achieves both zero-copy and full POSIX compatibility.
In fact, as we formally demonstrate in~\S\ref{sec:proof}, these goals are inherently conflicting under today’s mainstream OS architecture, therefore, full POSIX compatibility inevitably requires data copying.

\begin{table*}[htbp]
\centering
\resizebox{\textwidth}{!}{
\begin{tabular}{ll c cccc}
\toprule
\multirow{2}{*}{\textbf{Category}} & \multirow{2}{*}{\textbf{Representative Works}} & \multirow{2}{*}{\textbf{Zero-Copy}} & \multicolumn{3}{c}{\textbf{POSIX Compatibility}} \\
\cmidrule(lr){4-6}
& & & \textbf{Standard API} & \textbf{Flexible Buffering\textsuperscript{\ddag}} & \textbf{No Logic Modification\textsuperscript{\dag}} \\
\midrule
New syscall & Copier~\cite{copier}, TCP Splice~\cite{tcp_splice} & $\times$ & $\times$ & $\checkmark$ & $\checkmark$ \\
New proxy & Prism~\cite{prism}, QDSR~\cite{qdsr}, MireSGA~\cite{miresga} & $\times$ & $\checkmark$ & $\checkmark$ & $\times$ \\
Pre-registered memory & Arrakis~\cite{arrakis}, standard RDMA~\cite{rdma} & $\checkmark$ & $\times$ & $\times$ & $\checkmark$ \\
Memory remapping & Zero-copy socket~\cite{zerocopy-recv, zerocopy-send},
Solaris~\cite{solaris}, FreeBSD~\cite{freebsd}, SocksDirect~\cite{socksdirect} & $\checkmark$ & $\checkmark$ & $\times$ & $\checkmark$ \\
Specific Work & IX~\cite{IX} & $\checkmark$ & $\times$ & $\checkmark$ & $\checkmark$ \\

\textbf{Selective Copy} & \textbf{Libra} & $\times$ & $\checkmark$ & $\checkmark$ & $\checkmark$ \\
\bottomrule
\end{tabular}
}
\caption{Comparison of existing optimizations.}
\label{tab:related_work}

\noindent\raggedright
\textsuperscript{\ddag}\footnotesize
``Flexible Buffering'' allows applications to allocate arbitrarily sized buffers at any time, without pre-registration or alignment constraints.

\textsuperscript{\dag}\footnotesize
``No Logic Modification'' refers to whether the proxy preserves the standard forwarding semantics: receive request $\rightarrow$ forward request $\rightarrow$ receive reply $\rightarrow$ forward reply. 
``New Proxy'' designs break this model by forwarding only metadata and terminating the proxy's involvement early (e.g., receive request $\rightarrow$ forward request $\rightarrow$ termination), which constitutes a logical change in the proxy's operation—even if the application uses standard POSIX APIs.

\end{table*}

\sk{Add a table to better describe the meaning of zero-copy and compatibility, and related works. And the description here should be longer.}

Fortunately, we observe a critical behavioral trait of modern L7 proxies: they typically require only a tiny fraction of the data (the metadata, such as HTTP headers) to make routing and load-balancing decisions, while the massive bulk payload is never parsed or modified. This insight shifts our fundamental optimization paradigm. If we cannot avoid copying entirely, the logical solution is to \textbf{selectively copy} only the necessary metadata to the user space, while bypassing the redundant copy for the opaque payload. Instead, the bulk data is anchored in the kernel socket buffer and can be directly reused for subsequent transmission.

However, realizing this selective copy approach entails several key challenges:
(1) How can the kernel distinguish L7 metadata, selectively copy only the metadata and anchor the bulk payload?
(2) How can the kernel know when to identify metadata and when to skip it to synchronize with the proxy\footnote{L7 proxies operate as protocol-aware state machines: they parse metadata (e.g., HTTP headers) to determine the payload length, forward exactly that many bytes, and then await a new header—requiring the kernel to track this state to avoid misinterpreting payload as metadata.}?
(3) How can the anchored payload be safely retrieved for further transmission? 

\textbf{Our approach: Libra.}
We present \textit{Libra}, a near zero-copy framework that addresses these challenges by leveraging the OS principle of separating mechanism from policy: the kernel provides a safe, universal mechanism for selective metadata extraction and in-kernel payload anchoring, while lightweight eBPF programs—supplied by users—encode protocol-specific logic to separate L7 metadata needed for control decisions from bulk data-plane payloads. This design does not require application modifications or introduce special infrastructure requirements, while preserving standard POSIX socket semantics. (Details in \S\ref{sec:design}.)

We make the following contributions:

\hangindent=2em\hangafter=1 1. We establish that zero-copy and POSIX compatibility are fundamentally incompatible in mainstream operating systems. Given this constraint, we propose a paradigm shift from rigid full-payload bypassing to \textit{selective copy}. By accepting a negligible copy overhead for lightweight L7 metadata while safely bypassing the opaque bulk payload, Libra achieves near zero-copy performance—without requiring application modifications or hardware-specific constraints. 

\hangindent=2em\hangafter=1 2. We design and implement Libra, an OS-level framework that leverages eBPF to enable user-defined, protocol-aware L7 parsing logic in the kernel, while data forwarding is handled transparently by kernel functions. This ensures compatibility with unmodified proxy applications for variable L7 protocols.

\hangindent=2em\hangafter=1 3. We demonstrate performance gains for L7 proxies while preserving full application transparency and avoiding specialized infrastructure. Using standard network configurations, Libra accelerates unmodified Nginx and HAProxy deployments by up to 4.2x in plaintext throughput and reduces P99 tail latency by over 90\% versus the standard Linux stack. With hardware-offloaded kTLS, Libra doubles encrypted throughput and cuts tail latency by 65\%.

\textit{This work does not raise any ethical issues.}
\section{Background and Motivation}
\label{sec:background}

\subsection{Copies in Modern L7 Proxies}
\label{sec:copies_in_modern_l7_proxies}

Modern cloud-native infrastructures rely heavily on Layer-7 (L7) proxies for essential capabilities like request routing and load balancing. To process L7 protocols (e.g., HTTP), these proxies use the standard POSIX socket API. While kernel-bypass solutions (e.g., DPDK~\cite{dpdk}) exist, their operational complexity limits adoption in production~\cite{pegasus, joyride}; thus, production environments predominantly rely on the standard in-kernel network stack to maintain compatibility, which requires copying the entire payload between kernel and user space. This kernel-user data movement has become the primary scalability bottleneck.

To understand the severity and universality of this ``copy tax,'' we used Linux \texttt{perf} and flame graphs~\cite{perf2010,flamegraph} to analyze the proportion of CPU overhead in proxy processes spent on data copying during \texttt{recv} (kernel-to-user) and \texttt{send} (user-to-kernel) operations across four production-grade L7 proxies (Nginx, HAProxy, Traefik, and Envoy) under varying payload sizes (16KB and 256KB).

Figure~\ref{fig:copy_asymmetry} shows significant cross-boundary copy overhead in L7 proxies. In Nginx, copying consumes 15.9\% of CPU cycles for 16\,KB payloads and 55.6\% for 256\,KB. Other proxies also incur substantial overhead: HAProxy (39.4\%), Traefik (33.3\%), and Envoy (24.6\%).




Furthermore, profiling reveals a critical \textbf{Rx/Tx asymmetry}: receive-side copying consumes far more CPU cycles than transmit-side. For example, with 256KB payloads in Nginx, Rx copying accounts for 45.2\% of CPU cycles—over 4$\times$ the Tx overhead (10.4\%). To understand this, we separately profile TLB miss rates on the receive and transmit paths using a single-stream file download setup (client dominated by Rx, server by Tx), as shown in Figure~\ref{fig:tlb_breakdown}. On the Rx path, the NIC DMAs packets into discontiguous page fragments, and \texttt{copy\_to\_user} performs scattered reads across them, defeating prefetching and causing frequent page table walks. This results in a high TLB miss rate ($\sim$0.40\%). In contrast, the Tx path uses \texttt{copy\_from\_user} to pack data into pre-allocated, physically contiguous high-order pages (e.g., 32KB), minimizing page crossings. Consequently, for 512KB payloads, the Rx TLB miss overhead is 10$\times$ higher than Tx.

\textbf{Takeaway \#1:} Kernel--user copying is non-negligible in L7 proxies, and the fragmented Rx TLB miss overhead further exacerbates the bottleneck.

\begin{figure}[t]
    \centering
    \begin{subfigure}[b]{0.57\columnwidth}
        \centering
        \includegraphics[width=\linewidth]{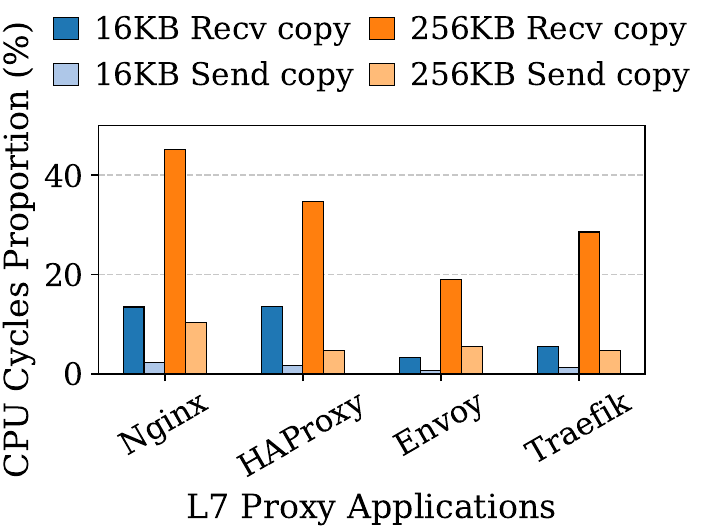}
        \caption{Proportion of total process CPU cycles spent on cross-boundary copies.}
        \label{fig:copy_asymmetry}
    \end{subfigure}
    \hfill
    \begin{subfigure}[b]{0.39\columnwidth}
        \centering
        \includegraphics[width=\linewidth]{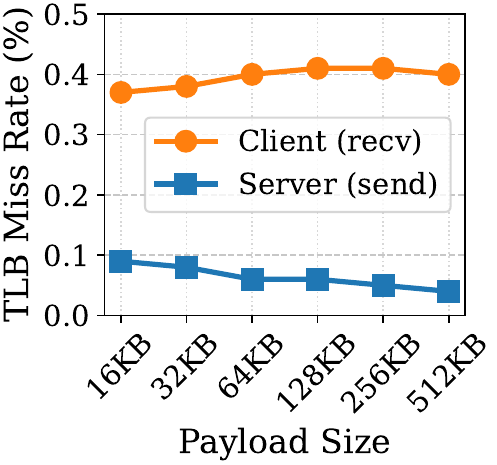}
        \caption{TLB miss rate under \texttt{recv}/\texttt{send} operations.}
        \label{fig:tlb_breakdown}
    \end{subfigure}
    \caption{Analysis of cross-boundary copying overhead and TLB efficiency.}
    \label{fig:combined_plots}
\end{figure}

\subsection{The Conflict between Zero-Copy and POSIX Compatibility}
\label{sec:proof}
As data copy overhead has become a major performance bottleneck, zero-copy mechanisms are highly desired. To formalize this discussion, we define \textit{zero-copy} as a data transfer mechanism where, once data resides in host memory, it is subsequently received or sent by the application without any memory-to-memory data movement
\footnote{
``Data movement'' includes hardware-assisted transfers such as DMA. For DRAM-to-DRAM copies, DMA incurs non-negligible setup overhead and often achieves lower throughput than optimized CPU-based copying, resulting in reduced efficiency~\cite{copier}.}. We further define \textit{POSIX compatibility} as the ability of a system to support unmodified standard applications that are written exclusively using standard POSIX APIs, without requiring any code changes to achieve expected behavior and semantics~\cite{posix}.

However, for mainstream operating systems, achieving zero-copy fundamentally conflicts with maintaining strict POSIX compatibility. In this section, we provide a formal proof to explain why this conflict exists.

\noindent\textbf{Definitions and Desired Properties:}

\noindent\textbf{• Payload ($N$):} The incoming network data payload.

\noindent\textbf{• Virtual Buffer ($V$):} The application-provided virtual memory buffer passed to \texttt{recv}.

\noindent\textbf{• NIC DMA Target ($P_{NIC}$):} The physical memory region written by the NIC via DMA when receiving payload $N$.

\noindent\textbf{• Initial Physical Mapping ($P_{init}$):} The physical page first mapped to virtual buffer $V$.
    
\noindent\textbf{• Standard \texttt{recv} Paradigm:} Upon returning from a \texttt{recv} call, the application's virtual buffer $V$ maps to a physical memory region $P$ such that $P$ contains $N$.
    
\noindent\textbf{• Zero-Copy (G1):} $N$ is DMAed into a unique physical region $P_{NIC}$ and is never duplicated. Therefore, any physical region $P$ containing $N$ must strictly be $P_{NIC}$ ($P \equiv P_{NIC}$).

\begin{figure}[htbp]
    \centering
    \includegraphics[width=0.8\columnwidth]{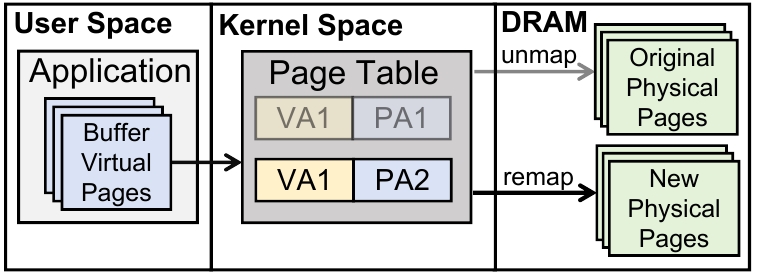}
    \caption{Virtual memory remapping process.}
    \label{fig:remap}
\end{figure}
    
\noindent\textbf{• POSIX-Compatible \texttt{recv} (G2):} The system supports the standard \texttt{recv(void *buf, size\_t len, ...)} interface. Under this paradigm, the buffer \texttt{buf} (virtual buffer $V$) follows the properties of \textit{Flexible Buffering}: it is allocated prior to the \texttt{recv} call, and its length \texttt{len} ($L$) can be set to any positive value by the application developer.

\textbf{Basic Assumptions (System Constraints):}

\noindent\textbf{• A1 (Finite NIC Buffering):} NICs have limited on-chip memory and must DMA incoming packets into \textit{pre-allocated} host physical memory immediately after hardware-level processing.
    
\noindent\textbf{• A2 (Page-Granular Remapping):} Applications access data via virtual addresses, which the OS maps to physical memory pages. The OS can update existing virtual-to-physical mappings only at page granularity. Consequently, if a virtual buffer $V$ is already backed by physical memory $P_{init}$, remapping $V$ to a new physical region $P_{new}$ requires $V$ to be page-aligned
\footnote{Page alignment means the buffer starts exactly at the beginning of a virtual memory page (e.g., at an address that is a multiple of 4 KB), and this requirement is not central to our proof.}
and its length $L$ to be a multiple of the page size; otherwise, modifying the page-table mappings would corrupt data associated with adjacent virtual addresses mapped to the same physical page.

\textbf{Virtual Memory Model.}
Applications access memory through virtual address ranges, which are mapped by the OS to physical memory at page granularity via page tables, as illustrated in Figure~\ref{fig:remap}. 
A virtual buffer $V$ corresponds to a contiguous virtual address range, whose physical backing is established by the OS (e.g., $VA1 \rightarrow PA1$). 
Once established, changing the underlying physical memory requires the OS to update the page table entries—unmapping the original physical pages ($PA1$) and remapping the virtual pages to point to new physical pages ($PA2$). 
Crucially, this hardware-level remapping operation is strictly constrained by page granularity (A2).


\begin{theorem}
For any operating systems satisfying A1 and A2, G1 and G2 cannot be satisfied simultaneously.
\label{thm:impossibility}
\end{theorem}

\textbf{Proof of Theorem:}
Assume the application allocates a virtual buffer $V$ for receiving data, whose initial physical backing is denoted $P_{init}$.
Based on A1, $N$ must be initially DMAed into some physical region $P_{NIC}$. There are only two possible targets for this placement:

\textbf{Case 1: $P_{NIC} = P_{init}$ (NIC DMA $\rightarrow$ application buffer’s physical page).} \\
If the NIC is to DMA the payload $N$ directly into the physical pages backing the virtual buffer $V$, a fundamental tension arises with G2, which permits the application to allocate $V$ no later than the \texttt{recv} call. Crucially, standard POSIX semantics allow the payload $N$ to arrive \textit{before} the buffer $V$ even exists, as illustrated below:

\vspace{0.5em}

\begin{lstlisting}[language=C,
    frame=single,
    basicstyle=\footnotesize\ttfamily,  
    columns=fullflexible,                
    keepspaces=true,                     
    breaklines=true,                     
    breakatwhitespace=true,              
    escapeinside={(*}{*)},
    xleftmargin=0.5em,                   
    framexleftmargin=0.5em]
// Connection established
int proxy_fd = accept(listening_fd, ...); 

// The incoming payload N may arrive during this delay,
// while being buffered in non-application memory.
sleep(1); 

// Buffer V is allocated after N has already arrived (G2).
void *buffer = malloc(1024); 

// N is successfully read into V.
recv(proxy_fd, buffer, sizeof(buffer), 0); 
\end{lstlisting}

Since $V$ and its physical backing do not yet exist at the moment of $N$'s arrival, $N$ has no pre-allocated destination for direct placement, thereby violating A1. Otherwise, to resolve this temporal paradox, the system would have to force the application to pre-register buffers before any data arrives, which violates G2.

\textbf{Case 2: $P_{NIC} \neq P_{init}$ (NIC DMA $\rightarrow$ non-application physical page).} \\
If the NIC DMAs $N$ into $P_{\text{NIC}}$, a region in non-application memory (e.g., a kernel page pool). The application execution flow is illustrated below:

\begin{lstlisting}[language=C,
    frame=single,
    basicstyle=\footnotesize\ttfamily,  
    columns=fullflexible,                
    keepspaces=true,                     
    breaklines=true,                     
    breakatwhitespace=true,              
    escapeinside={(*}{*)},
    xleftmargin=0.5em,                   
    framexleftmargin=0.5em]
// 1. Buffer Allocation (G2): Buffer virtual address V 
// is reserved. L: Length of buffer (arbitrary)
void *buffer = malloc(L);

// 2. Intermediate Application Logic:
// Any operation (e.g., initialization or
// access to adjacent data) that causes V to be backed by
// physical pages P_init (V (*$\rightarrow$*) P_init) before recv().
do_application_logic();

// 3. Remap Attempt during recv()
// Goal: change mapping of V from P_init to P_NIC
// Remapping V (*$\rightarrow$*) P_NIC requires L % PageSize == 0
recv(fd, buffer, L, 0); 
\end{lstlisting}

As shown in the flow above, the application allocates $V$ with an arbitrary length $L$ (G2), and any intermediate application logic causes $V$ to already be backed by physical pages $P_{init}$. According to the Standard \texttt{recv} Paradigm, the OS must ensure that $V$ maps to a physical region $P$ containing $N$. Furthermore, to satisfy G1, since $N$ cannot be duplicated, $P$ must be exactly $P_{NIC}$. Consequently, the OS is forced to establish the remapping $V \rightarrow P_{NIC}$. 

However, under A2, because $V$ is already backed by $P_{init}$, this hardware remapping is impossible if $L$ is not a multiple of the page size. Since $L$ can be arbitrary (G2), it is not guaranteed to be a page multiple. To resolve this conflict, the OS must either copy $N$ to $V$'s existing physical backing $P_{init}$ (violating G1) or force the application to use strictly page-multiple buffers (violating G2). 

\textbf{Why Dedicated Pages Don’t Solve the Problem:}
One might transparently intercept all user buffer allocations to enforce page-aligned, page-sized blocks, ensuring each buffer exclusively occupies its physical pages—avoiding adjacent corruption even when $L$ is not page-aligned. Yet writing beyond $[buf, buf+L)$ during receive remains an out-of-bounds write, violating both the ISO C and POSIX standards~\cite{isoc, posix}. Forcibly padding every buffer to page granularity also causes severe internal fragmentation, rendering the approach impractical for general use.

\textbf{Takeaway \#2:} Zero-copy and POSIX compatibility are fundamentally incompatible under current OS abstractions.

\subsection{Existing Optimizations}
\label{sec:existing_optimizations}

Constrained by this limitation, all existing zero-copy optimizations fail to simultaneously achieve zero-copy and preserve compatibility. We analyze these existing efforts and categorize them as follows:

\textbf{New proxy mechanisms.}
Solutions like Prism~\cite{prism} and QDSR~\cite{qdsr} bypass the proxy for bulk data using Direct Server Return (DSR)~\cite{dsr, ananta}, while Miresga~\cite{miresga} achieves proxy bypass via programmable switches. Instead of the standard four-step flow—receive request, forward to server, receive reply, send to client—the proxy forwards the request and terminates its involvement, letting the server reply directly to the proxy's downstream (e.g., the client). This fundamentally alters the proxy’s role, requires source-code changes, and breaks compatibility. It also introduces severe cross-node coupling: DSR mandates backend modifications, while hardware offloading requires specific network devices.

\textbf{New system call mechanisms.}
Solutions such as Copier and TCP Splice optimize data transfer by introducing custom system calls~\cite{copier, tcp_splice}. Specifically, Copier replaces content-agnostic \texttt{recv} and \texttt{send} operations with new system calls that delegate them to a kernel thread, merging them into a single asynchronous copy. TCP Splicing, on the other hand, provides a mechanism to splice two sockets, enabling the kernel to forward subsequent packets between them entirely in kernel space. To leverage these mechanisms, the application must explicitly invoke the copy-bypass function, rather than using the standard API. This necessitates modifications to the application to accommodate the new system calls.
\sk{<- hard to understand.}  

\textbf{Zero copy mechanism.}
Arrakis~\cite{arrakis} uses hardware virtualization (e.g., SR-IOV~\cite{sriov}) to DMA packets directly into user-accessible memory. Similarly, RDMA~\cite{rdma} offloads the entire transport stack to hardware, allowing the NIC to DMA data straight into application buffers. Both require applications to pre-register memory with the NIC, necessitating intrusive code changes that break compatibility.

Other approaches achieve zero-copy via virtual memory remapping, such as Solaris~\cite{solaris}, FreeBSD~\cite{freebsd}, and Linux zero-copy sockets~\cite{zerocopy-recv,zerocopy-send}, which deliver packets by page remapping. SocksDirect~\cite{socksdirect} similarly preserves POSIX semantics for RDMA by receiving into a kernel page pool and remapping pages to user space—yet its zero-copy path requires application buffers to be page-aligned and sized to page multiples. In fact, all page-remapping schemes share this limitation: for arbitrarily sized buffers (the common case in our proof), they inevitably fall back to traditional copying. Moreover, these approaches require NIC support for header-data split and often mandate MTU adjustment~\cite{header_data_split}. 

Additionally, custom dataplanes like IX~\cite{IX} deliver raw packet pointers directly to applications, but force proxies to handle scattered fragments via non-POSIX APIs.

In summary, all existing approaches exhibit trade-offs between zero-copy and POSIX compatibility, consistent with our theorem.

\subsection{Motivation: L7 Proxies Often Require Only Metadata}
\label{sec:motivation}

Driven by the inherent conflict demonstrated in \S\ref{sec:proof}, we argue that physical memory copying cannot be completely eliminated while preserving strict POSIX compatibility. However, this insight suggests a shift in perspective: rather than trying to eliminate copying altogether, we should ask whether the proxy actually needs to copy the \textit{entire} payload.

To answer this question, we analyze the memory access patterns and data utilization behavior of L7 proxies. While some proxies must perform deep inspection on incoming client requests for security filtering or routing, HTTP workloads are overwhelmingly read-heavy~\cite{read_heavy}. In the dominant server-to-client response path, the proxy often simply forwards the data. For these response flows, the proxy only needs to parse a minimal amount of L7 metadata to update its internal state or manage connection mappings~\cite{copier, miresga}.

To quantify this feature, we analyze what exactly constitutes ``metadata'' for the proxy across different generations of the HTTP protocol, which dominates L7 proxy traffic, as summarized in Table~\ref{tab:http_metadata}. Each HTTP version structures messages differently; we summarize their formats below:

\noindent\textbf{HTTP/1.0 and HTTP/1.1:}  
An HTTP message comprises headers followed by a contiguous body~\cite{http1.0}. In HTTP/1.1, when chunked transfer encoding is used, the body is split into arbitrarily sized chunks, each prefixed by a short header (typically $\leq$10 bytes)~\cite{http1.1}.

\noindent\textbf{HTTP/2:}  
Messages are segmented into binary frames: a \texttt{HEADERS} frame (with HPACK-compressed headers) followed by one or more \texttt{DATA} frames (carrying the payload). Each frame begins with a 9-byte header and is at most 16 KB~\cite{http2}.

\noindent\textbf{HTTP/3:}  
Messages are framed over QUIC streams: a \texttt{HEADERS} frame (with QPACK-compressed headers)~\cite{qpack, quic}, followed by \texttt{DATA} frames, each with a small frame header and typically up to 16 KB~\cite{http3}.

\begin{table}[htbp]
\centering
\resizebox{\columnwidth}{!}{%
\begin{tabular}{@{}lp{4.5cm}p{3.5cm}@{}}
\toprule
\textbf{Protocol} & \textbf{Proxy-Required Metadata} & \textbf{Typical Length} \\ \midrule
\textbf{HTTP/1.0} & 
HTTP header 
& Hundreds of Bytes \\ \addlinespace

\textbf{HTTP/1.1} & 
HTTP header\newline 
Chunked Encoding delimiter (optional) & 
Hundreds of Bytes \newline 
$\sim$5--10 Bytes per chunk\\ \addlinespace

\textbf{HTTP/2} & Frame header\newline
HPACK HTTP header & 
9 Bytes per frame\newline 
Few--Hundreds of Bytes \\ \addlinespace

\textbf{HTTP/3} & 
Frame header\newline 
QPACK HTTP header & 
$\sim$2--8 Bytes per frame\newline 
Few--Hundreds of Bytes \\ \bottomrule
\end{tabular}%
}
\caption{HTTP Metadata Across Versions.}
\label{tab:http_metadata}
\end{table}

As shown in Table~\ref{tab:http_metadata}, the metadata a proxy actually needs to access is typically only a few bytes to a few hundred bytes. In contrast, bulk payloads are entirely unparsed and commonly range from 2 to 3\,MB in size~\cite{httparchive}. To the proxy, such massive payloads are merely an \textit{opaque byte stream} forwarded to the client.

Based on this observation, we propose \textit{selective copy}: an ideal data-path architecture for protocols with an in-kernel transport layer. On packet arrival, the kernel copies only the essential L7 metadata to user space, enabling the proxy to perform routing logic. The bulk payload remains in kernel memory and is later forwarded directly to the NIC upon send, achieving zero-copy for the data path. This selective approach still does not break our Theorem~\ref{thm:impossibility}. It achieves \textit{near zero-copy} performance across the end-to-end proxying workflow, incurring only a negligible copy of L7 metadata.

\textbf{Takeaway \#3:} L7 proxies only need to inspect a tiny fraction of data (metadata); bulk payloads are opaque and forwarded unchanged.

\subsection{Challenges and Solutions}
\label{sec:challenges_and_solutions}

To realize the ideal selective copy framework, the system must precisely intercept and manipulate the data path without violating standard socket semantics. However, this goal gives rise to the following three architectural challenges:

\noindent\textbf{C1: Identifying L7 metadata in the kernel.}  
To physically avoid copy overhead, the boundary of L7 metadata must be identified inside the kernel, immediately before physical data movement occurs. However, the OS operates at the transport layer (L4) and treats all TCP payloads as an undifferentiated byte stream. Since L7 protocols are diverse, variable in length, and dictated by proxy-specific configurations, hardcoding static L7 parsers into the monolithic kernel is fundamentally inflexible and impractical.

\noindent\textbf{S1: User-Programmable Policy Injection.}  
To bridge this gap, metadata parsing must be decoupled as a user-defined \textit{policy}. Libra leverages eBPF~\cite{ebpf,ebpf2} to safely inject protocol-specific parsing logic into the kernel's data path, enabling developers to dynamically define L7 metadata boundaries for varying workloads without modifying the kernel.

\noindent\textbf{C2: Executing Selective Copies without Breaking Proxy State.} In a typical L7 proxy, after parsing the metadata to learn the payload size, the proxy expects to forward the bulk data and then parse the next message’s metadata. If the kernel blindly skips copying all subsequent bytes after the first metadata, the next message’s metadata will also be bypassed—causing the proxy to crash or hang. Thus, selective copying must track the proxy’s L7 protocol state.

\noindent\textbf{S2: State-Driven Kernel Paradigm via eBPF State Machines.} Libra solves this with an in-kernel eBPF state machine that mirrors the proxy’s L7 logic. It parses metadata, computes payload boundaries, and dynamically controls the data path: it enables user-space copying for metadata and switches to zero-copy for the payload. Crucially, \texttt{recv} always returns the full logical message length even when only metadata is copied—ensuring transparency. Once the payload ends, eBPF resets the state for the next metadata.

\begin{figure*}[htbp]
    \centering
    \begin{subfigure}[b]{0.45\linewidth} 
        \centering
        \includegraphics[width=\linewidth]{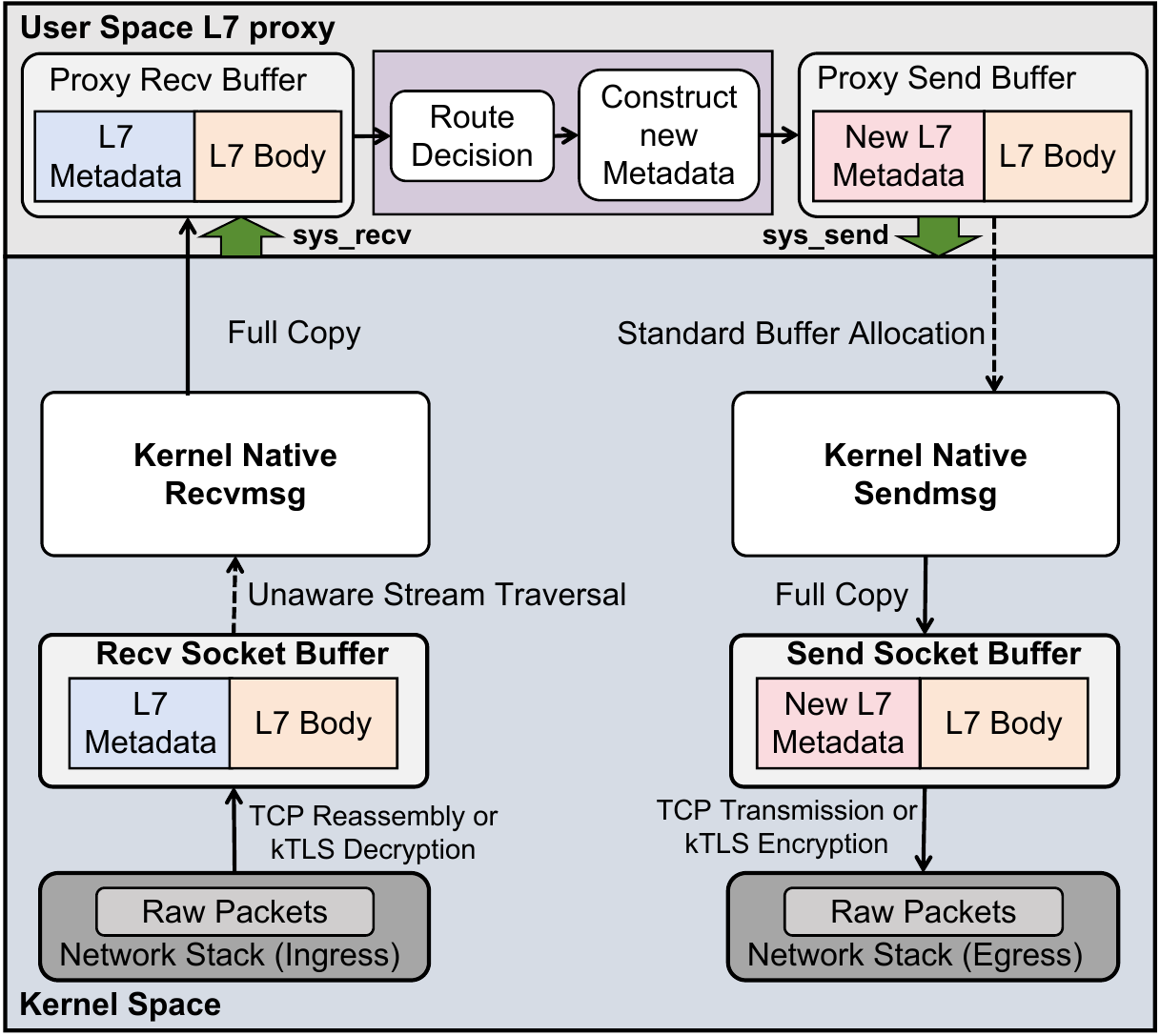}
        \caption{Overview of standard baseline.}
        \label{fig:arch_baseline}
    \end{subfigure}
    \hspace{2em}
    \begin{subfigure}[b]{0.45\linewidth}
        \centering
        \includegraphics[width=\linewidth]{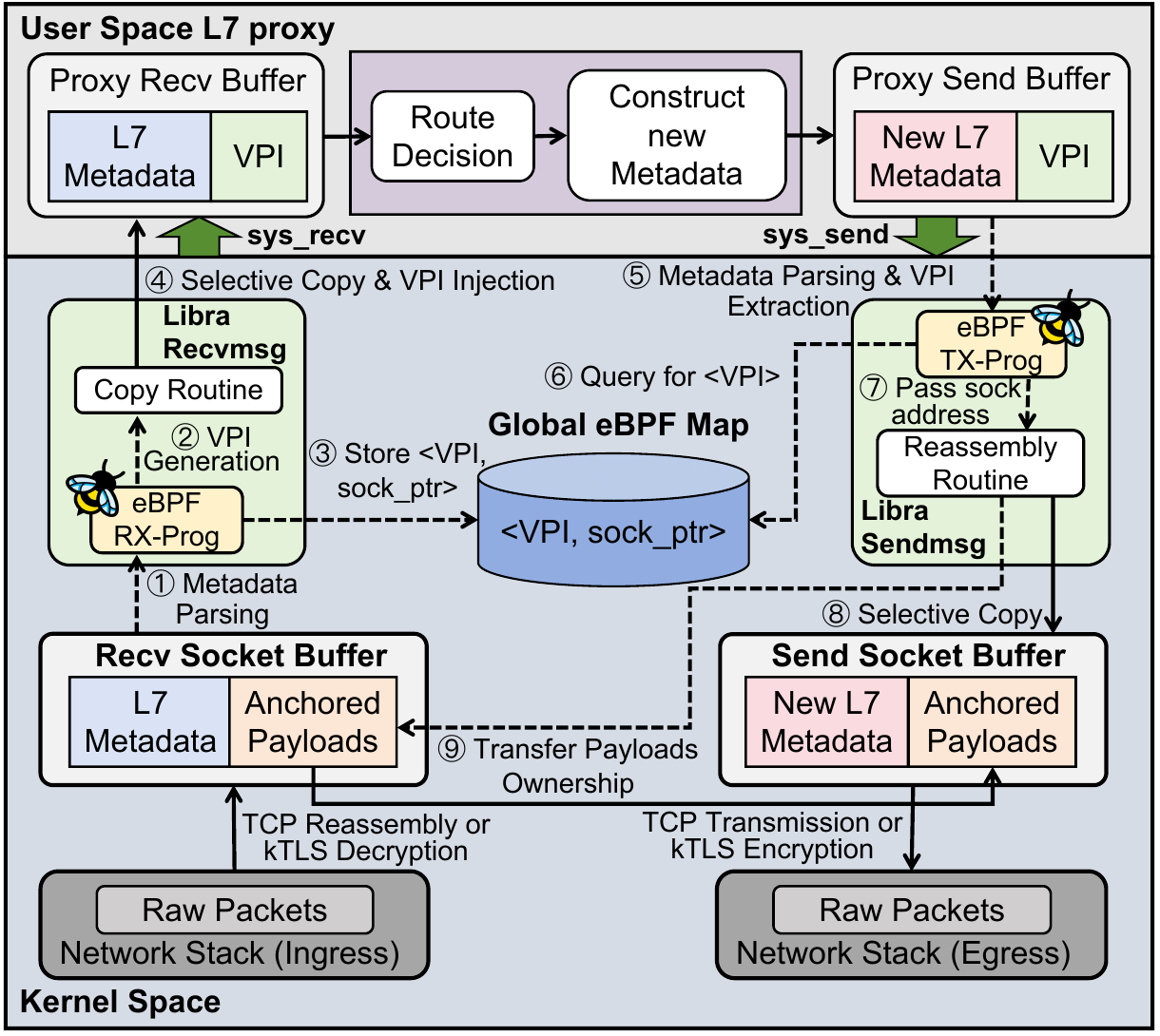}
        \caption{Overview of Libra.}
        \label{fig:arch_Libra}
    \end{subfigure}
    
    \caption{Data flow comparison between standard and Libra-based L7 proxy forwarding.}
    \label{fig:arch_overview}
\end{figure*}

\noindent\textbf{C3: Forwarding Anchored Payloads Across Sockets.}  
If the bulk payload is anchored in the kernel during receive, a problem arises at transmission: the proxy uses separate TCP connections for the client and backend, and their association exists only in user space. The kernel thus cannot map anchored payloads to the outgoing socket.

\noindent\textbf{S3: Virtual Payload Identifier (VPI).}  
To bridge the decoupled connections, Libra introduces the Virtual Payload Identifier (VPI). The Rx eBPF program replaces the payload with a small VPI in the user buffer; the proxy forwards it as opaque data via \texttt{send}. The Tx eBPF program intercepts the VPI, resolves it to the anchored kernel memory, and performs zero-copy forwarding.

\section{Design}
\label{sec:design}

\subsection{Libra Workflow}
\label{sec:architecture}

To achieve near zero-copy performance for L7 proxies while preserving POSIX compatibility, we present Libra. It separates user-defined L7 policy from in-kernel memory mechanisms by deploying two eBPF programs—\texttt{RX-Prog} and \texttt{TX-Prog}—that run a state machine synchronized with the proxy’s L7 lifecycle. This state machine dynamically identifies metadata–payload boundaries, enabling state-aware kernel primitives to physically decouple them.

Figure~\ref{fig:arch_overview}(a) illustrates the standard POSIX I/O path, where every \texttt{recv} operation triggers a full copy of the entire payload (including opaque application data) into the user space and vice versa. As shown in Figure~\ref{fig:arch_overview}(b), Libra fundamentally redesigns this datapath through two phases: Ingress anchoring and Egress reassembly.

\noindent \textbf{Ingress: Selective Copy and Payload Anchoring.} 
When the proxy issues a \texttt{recv} system call to ingest data, the \texttt{RX-Prog} is triggered from within the kernel's receive copy routine. It inspects the stream to perform L7 metadata parsing (\textcircled{1}). Once parsed, Libra performs a \textit{selective copy}, delivering only the L7 metadata to the proxy's user-space buffer. The massive opaque payload is deliberately left anchored within the kernel's receive socket buffer.

To bridge ingress and egress paths without exposing kernel addresses, \texttt{RX-Prog} generates a 64-bit Virtual Payload Identifier (VPI) (\textcircled{2}), a secure, opaque identifier to track the source socket where the original payload remains anchored. It stores the mapping \texttt{<VPI, sock\_ptr>} in a global eBPF hash map (\textcircled{3}) and injects the VPI into user space in place of the payload (\textcircled{4}). The unmodified proxy treats the VPI as a normal byte stream, while \texttt{TX-Prog} later uses it to retrieve the socket context and perform zero-copy forwarding.

\noindent \textbf{Egress: Anchored Payload Retrieval and Stream Reassembly.} 
In the user space, the proxy processes metadata, executes routing logic, and constructs a new L7 header. On \texttt{send}, it blindly writes the header followed by the unmodified VPI back to the kernel.

On the egress datapath, \texttt{TX-Prog} intercepts the stream in the transmit copy routine, parses the new metadata, and extracts the embedded VPI (\textcircled{5}). It queries the global eBPF map with this VPI (\textcircled{6}) to retrieve the original receive socket’s \texttt{sock\_ptr} (\textcircled{7}), then copies only the newly generated metadata from user space (\textcircled{8}). Finally, Libra performs a zero-copy pointer manipulation—transferring payload ownership (\textcircled{9})—to combine the metadata with the kernel-anchored payload, completing reassembly and forwarding.

To orchestrate this workflow safely and efficiently, Libra relies on three core components: VPI tracking mechanism, ingress decoupling logic, and egress reassembly routine, which we detail in the remainder of this section. Mechanisms ensuring safety guarantees and performance optimizations are deferred to Appendix~\ref{sec:others}.

\subsection{The VPI: Tracking Anchored Payloads}
\label{sec:design_vpi}
As established in \S\ref{sec:challenges_and_solutions}, tracking anchored payloads across disconnected receive and transmit sockets is critical to reducing copies. Libra bridges this gap with the \textbf{Virtual Payload Identifier (VPI)}, a secure, position-independent abstraction that remains opaque to the application.

\noindent\textbf{Position-Independent Tracking.} 
The VPI is a 64-bit logical placeholder embedded in the user-space buffer. To an unmodified proxy, it is indistinguishable from payload bytes, surviving any internal buffer shuffling or \texttt{memcpy}. When the proxy issues \texttt{send}, the VPI is included in the outbound stream, allowing Libra to maintain the tracking chain.

\noindent\textbf{Secure Mapping.} 
To prevent leaking kernel memory layouts and violating Kernel Address Space Layout Randomization (KASLR)~\cite{kaslr}, the VPI is never exposed as a raw physical or virtual pointer. Instead, it is generated as a cryptographically secure hash, serving as an $O(1)$ lookup key in a global eBPF map (\texttt{<VPI, struct sock*>}). 

\noindent\textbf{Safety Threshold and Admission Policy.} 
To ensure data integrity, Libra enforces a \textit{length-based admission policy}. Since the VPI occupies 8 bytes, selective zero-copy is only activated for opaque payloads $\geq 8$ bytes. Smaller payloads—where the copy overhead is negligible—are naturally handled via the standard full-copy path. This prevents the identifier from being truncated or misinterpreted across buffer boundaries, providing a robust fallback mechanism while maximizing performance for bulk data transfers.

\subsection{Ingress: State-Driven Selective Copy}
\label{sec:design_ingress}

To avoid the receive-side copy tax without breaking the proxy's L7 event loop, Libra's ingress datapath relies on a synchronized state machine, as illustrated in Figure~\ref{fig:ingress_state_machine}. The eBPF \texttt{RX-Prog} tracks the state to dictate memory operations performed in Libra's \texttt{recvmsg}. The lifecycle of an L7 message is managed across four states:

\noindent\textbf{1. DEFAULT: Metadata Parsing and Safety Fallback.}  
Each message starts in \texttt{DEFAULT}. \texttt{RX-Prog} inspects the socket’s receive queue via a bounded lookahead window (256 bytes, configurable via eBPF) from the current read offset. It uses the Knuth–Morris–Pratt (KMP) algorithm for deterministic $O(N)$ metadata parsing~\cite{kmp}. Once the L7 metadata is parsed and the body length determined, messages with payloads under 8 bytes remain in \texttt{DEFAULT} and follow the native full-copy path to avoid VPI overhead.

\noindent\textbf{2. METADATA\_PARSED: Selective Copy.}  
If metadata is parsed and the body is at least 8 bytes, \texttt{RX-Prog} checks whether the user buffer has space for the 8-byte VPI right after the metadata. If not, it transitions to \texttt{METADATA\_PARSED}, where the kernel copies only the metadata and defers VPI injection until a subsequent \texttt{recv} provides sufficient space.

\noindent\textbf{3. WRITE\_VPI: VPI Injection.}  
Once adequate buffer space is available, the state advances to \texttt{WRITE\_VPI}. In this state, if the metadata hasn’t been copied
\footnote{The eBPF state machine is evaluated atomically within a single \texttt{RX-Prog} / \texttt{TX-Prog} execution. Thus, transitions such as \texttt{DEFAULT} $\rightarrow$ \texttt{METADATA\_PARSED} $\rightarrow$ \texttt{WRITE\_VPI} can occur in one \texttt{recv} call if sufficient buffer space is available immediately. The kernel’s data-plane action is determined by the final state reached.},
the eBPF program instructs the kernel to copy it; regardless, it generates a VPI, stores the \texttt{<VPI, sock\_addr>} mapping in the global eBPF map, and injects the VPI right after the metadata, leaving the remaining buffer space intact. 
To preserve transparency, the kernel adjusts the return value of \texttt{recv} to reflect the logical message length (metadata + anchored payload), capped at the user-requested size.

\noindent\textbf{4. FAST\_PATH: Zero-Copy Consumption.}  
On subsequent \texttt{recv} calls for the same message, the connection enters the \texttt{FAST\_PATH} state. When the proxy attempts to read the bulk body, \texttt{RX-Prog} instructs the kernel to skip all physical data copies. The kernel merely advances the logical read offset to satisfy the requested length and returns the logical read value, while the payload remains physically anchored. The ingress state machine stays in this fast path until the \texttt{TX-Prog} (§\ref{sec:design_egress}) confirms that the payload has been fully transmitted and resets the state to \texttt{DEFAULT}.

\subsection{Egress: Payload Reassembly with L7 State Synchronization}
\label{sec:design_egress}

The egress data path is responsible for reassembling newly constructed L7 metadata with the anchored payload transparently. We introduce a state-driven, two-phase eBPF orchestration mechanism consisting of Pre-Send and Post-Send around Libra’s \texttt{sendmsg}, as illustrated in Figure~\ref{fig:egress_state_machine}.

\begin{figure}[t]
\centering
\includegraphics[width=0.9\columnwidth]{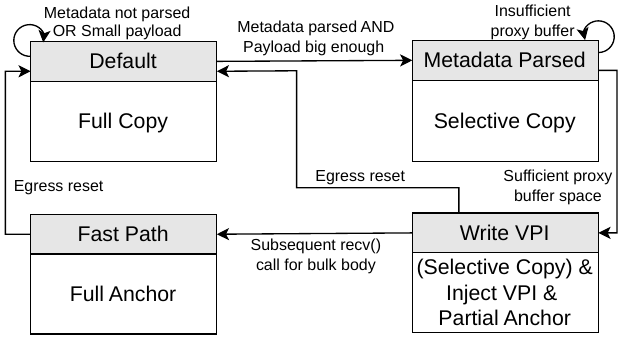}
\caption{%
  {The eBPF-driven ingress state machine.} \\
  \normalfont\small\emph{The top half denotes the eBPF control plane's logical phase, while the bottom half denotes the data plane kernel actions.}%
}
\label{fig:ingress_state_machine}
\end{figure}

\noindent\textbf{1. Pre-Send: Metadata Parsing and VPI Extraction.}  
On every \texttt{send} call, the Pre-Send eBPF hook \texttt{TX-Prog} intercepts the execution flow to infer the proxy’s logical intent. In the \texttt{DEFAULT} state, \texttt{TX-Prog} parses the outgoing user buffer to extract the L7 metadata, and performs dual length validation:  
(1) If the body is under 8 bytes, the eBPF state machine resets to \texttt{DEFAULT}, treating the transaction as ordinary short messages.
(2) If the body is sufficiently long but the available send length following the metadata contains fewer than 8 bytes, \texttt{TX-Prog} transitions to the \texttt{METADATA\_PARSED} state and defers VPI extraction until the next call.
(3) Only when both the identified body length and the available send length after the metadata are at least 8 bytes does the program extract a complete 8-byte Virtual Payload Identifier (VPI) immediately following the metadata.

Upon extracting a full 8-byte VPI, the program queries the global eBPF map. On a cache \textit{miss}—indicating no anchored payload—the system enters \texttt{FALLBACK\_BYPASS}, bypassing further L7 parsing to avoid KMP overhead and routing data via the native full-copy path. On a cache \textit{hit}, it retrieves the target \texttt{recv\_sock} and transitions directly to \texttt{FAST\_PATH}.

\noindent\textbf{2. Kernel Data Plane: Selective Copy and Zero-copy Reuse.}  
Guided by the final state of the Pre-Send eBPF program, the kernel executes the corresponding data-plane action per \texttt{send} call.  
In \texttt{METADATA\_PARSED}, the kernel copies only the L7 metadata into the \texttt{send\_sock}’s write queue.  
In \texttt{FAST\_PATH}, it first copies the uncopied metadata, then uses the retrieved \texttt{recv\_sock} pointer to transfer ownership of the anchored \texttt{sk\_buff}s into the egress stream, bypassing user space for the bulk payload.

\noindent\textbf{3. Post-Send: Application State Alignment.} 
Modern proxies often use non-blocking sockets, so a \texttt{send} call may transmit only part of the requested data. While Libra’s zero-copy mechanism ensures atomic delivery in \texttt{FAST\_PATH} (\S\ref{sec:send_sock}), \texttt{FALLBACK\_BYPASS} relies on the standard stack and remains subject to partial sends. To ensure correctness across both paths, the Post-Send eBPF hook runs after the kernel data-plane routine completes, recording the \textit{actual} transmitted bytes and updating a stateful cumulative counter—\textit{in all states except \texttt{DEFAULT}}\footnote{The \texttt{FALLBACK\_BYPASS} state is transient: once the current HTTP message is fully transmitted, Libra re-enables optimization for subsequent messages on the same connection. Hence, even in fallback mode, we track the cumulative send progress to correctly detect message completion and restore the optimized path.}.

Once the cumulative byte count matches the full message length (e.g., from \texttt{Content-Length}), the transmission is deemed complete. The Post-Send \texttt{TX-Prog} detects the request completion and triggers a synchronous cross-data-path cleanup—deleting the VPI map entry and resetting both TX and RX eBPF state machines to the \texttt{DEFAULT} state. This mechanism ensures perfect synchronization between the kernel’s internal state and the proxy’s request lifecycle, preparing the system for the next request over the reused connection.

\begin{figure}[t]
\centering
\includegraphics[width=0.9\columnwidth]{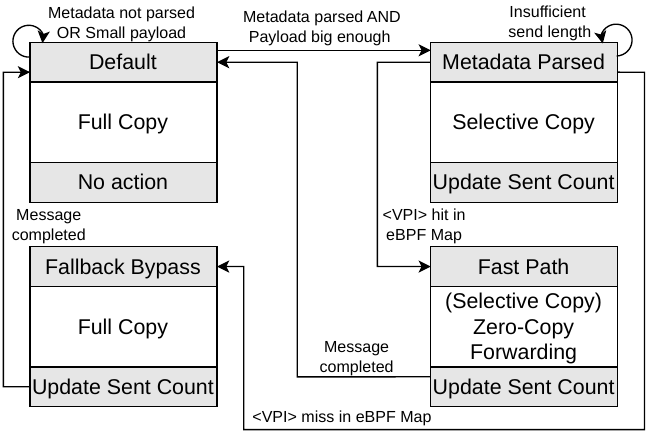}
\caption{%
  {The eBPF-driven egress state machine.} \\
  \normalfont\small\emph{The top part denotes the Pre-Send eBPF control plane logic; the middle part indicates the data plane action executed by the kernel; and the bottom part denotes the Post-Send eBPF control plane logic.}%
}
\label{fig:egress_state_machine}
\end{figure}

\subsection{Integration with Kernel TLS}
\label{sec:design_tls}

As TLS/SSL is now essential for secure communication, supporting encrypted streams is essential for any modern zero-copy architecture. Libra integrates with Kernel TLS (kTLS)~\cite{ktls}. By offloading symmetric crypto from user-space libraries (e.g., OpenSSL~\cite{openssl}) into the kernel, kTLS not only eliminates context switches but also exposes decrypted plaintext directly to our instrumented routines (see Figure.~\ref{fig:arch_Libra}).

\noindent\textbf{1. Ingress Visibility.}  
For inbound HTTPS traffic, kTLS decrypts the ciphertext before copying data to the user space. Since \texttt{RX-Prog} intercepts the stream after decryption, it operates identically to the plaintext case.

\noindent\textbf{2. Egress Encryption and Splicing.}  
To preserve zero-copy integrity during encryption, Libra adapts to both execution modes of kTLS:

\noindent\textbf{Software Mode:} To avoid copying plaintext, our specialized routine maps the pages of the anchored payload into a \textit{scatter-gather list} (i.e., \texttt{bvec}). This list is submitted directly to the kernel’s crypto engine for in-place encryption.

\noindent\textbf{Hardware Offload Mode:} For NIC-offloaded TLS, Libra performs a direct ownership transfer of the anchored \texttt{sk\_buff}s. Our logic dynamically re-frames each \texttt{sk\_buff} by carving out the necessary headroom and tailroom required by the NIC to inline-append TLS headers and cryptographic authentication tags (e.g., GHASH in AES-GCM).

This integration ensures that Libra’s performance benefits extend to encrypted web traffic.

\section{Evaluation}
\label{sec:evaluation}
In this section, we aim to demonstrate Libra's effectiveness in accelerating L7 proxies and to validate its architectural benefits. Specifically, we seek to answer the following core questions:

\noindent\textbf{• Q1 (Real-world Application Impact):} Can Libra accelerate unmodified L7 proxies, including those using kTLS, and match kernel-bypass systems like DPDK?

\noindent\textbf{• Q2 (State-of-the-Art Comparison):} How does Libra compare against the state-of-the-art baseline?

\noindent\textbf{• Q3 (Micro-architectural Efficiency):} What are the fundamental physical benefits (e.g., TLB cache preservation) of Libra, and what is the exact overhead of it?

We evaluate Libra on a three-node testbed consisting of a client, a proxy, and a backend server. Each machine is equipped with an Intel Xeon Silver 4110 processor~\cite{intel_xeon_4110} operating at 2.10 GHz and 128 GB of DDR4 memory. The nodes are connected via a 100 Gbps network. The proxy node is equipped with a Mellanox ConnectX-6 Dx (integrated in BlueField-2)~\cite{cx6dx}. The client and server nodes are equipped with Mellanox ConnectX-5~\cite{cx5}. 

The backend runs Nginx with \texttt{sendfile} enabled~\cite{sendfile} and is pinned to 4 CPU cores; client workloads are generated by \texttt{wrk}~\cite{wrk}. The proxy—whether Libra or baseline—is pinned to a single CPU core. By default, the proxy node runs Linux kernel v6.11.0, while client and server nodes run v4.15.0-213-generic. For state-of-the-art comparisons, the proxy node uses v5.15.131. We retain all default OS and NIC settings (e.g., MTU, hardware offloads) unless otherwise noted. Libra is implemented as a dynamically linked library and injected via \texttt{LD\_PRELOAD} to intercept socket system calls without application modifications~\cite{ld_preload}.

\noindent\textbf{Baselines.}
We compare Libra with the standard Linux network stack, the DPDK-based F-Stack~\cite{fstack}, and Copier~\cite{copier}. We specifically select Copier as our primary state-of-the-art comparison because it is a leading in-kernel framework that optimizes data copying without hardware dependencies and with minimal application changes. Many high-performance alternatives require infrastructure modifications or intrusive logic updates. In contrast, Copier operates entirely within the standard OS-level proxy datapath, aligning closely with our design goals and enabling a fair comparison under realistic deployment conditions.

\begin{figure*}[t]
\centering

\begin{subfigure}{0.49\textwidth}
  \centering
  \includegraphics[width=\linewidth]{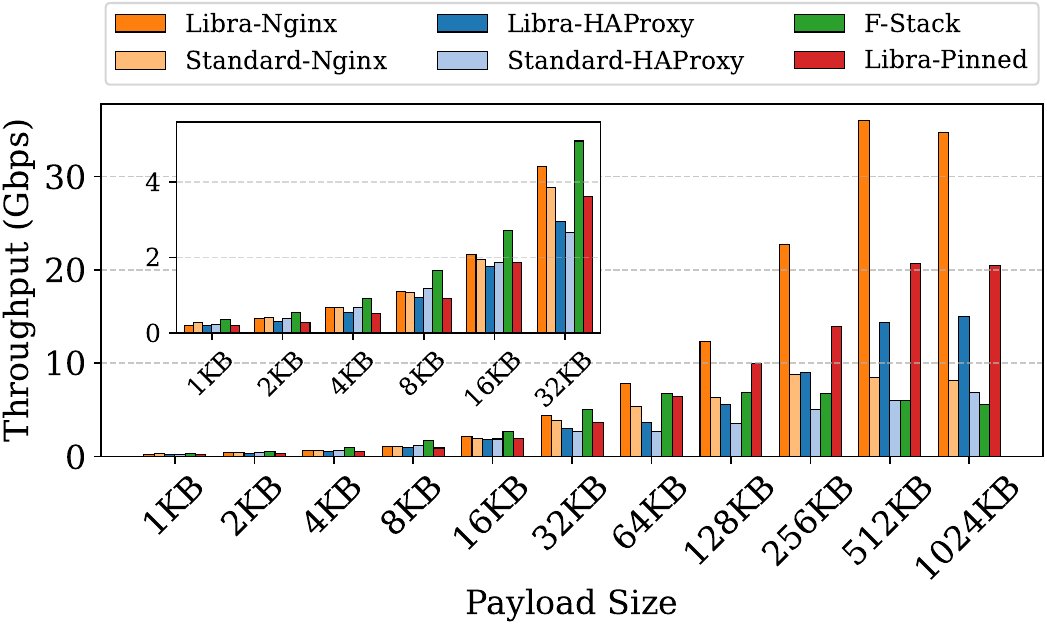}
  \caption{Nginx/HAProxy: Throughput}
  \label{fig:throughput_barchart}
\end{subfigure}
\hfill
\begin{subfigure}{0.49\textwidth}
  \centering
  \includegraphics[width=\linewidth]{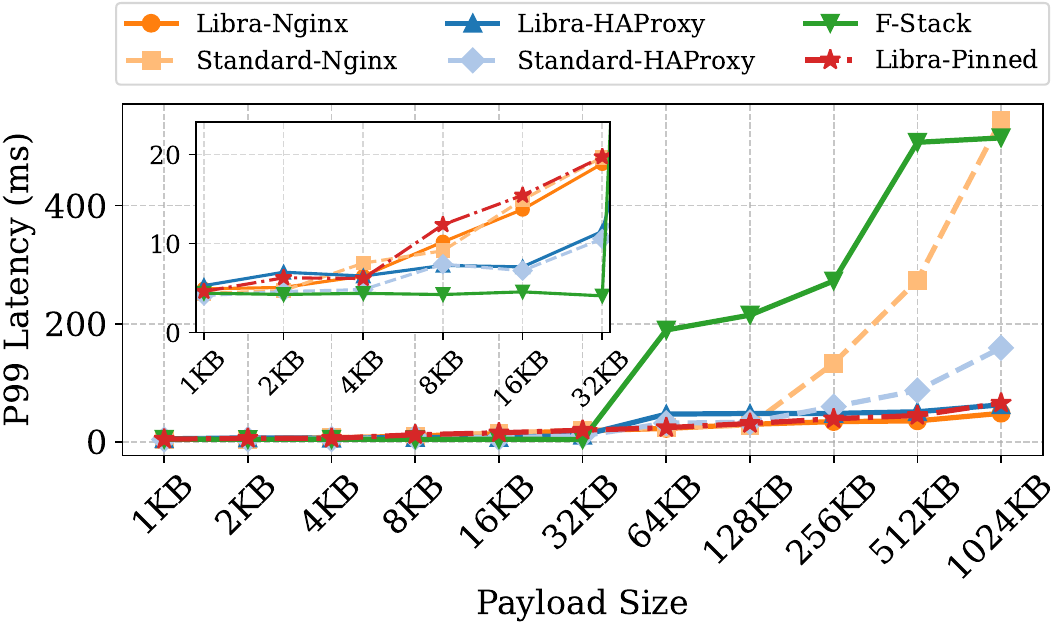}
  \caption{Nginx/HAProxy: P99 Latency}
  \label{fig:latency_p99_linechart}
\end{subfigure}

\vspace{0.3em} 

\begin{subfigure}{0.32\textwidth}
  \centering
  \includegraphics[width=\linewidth]{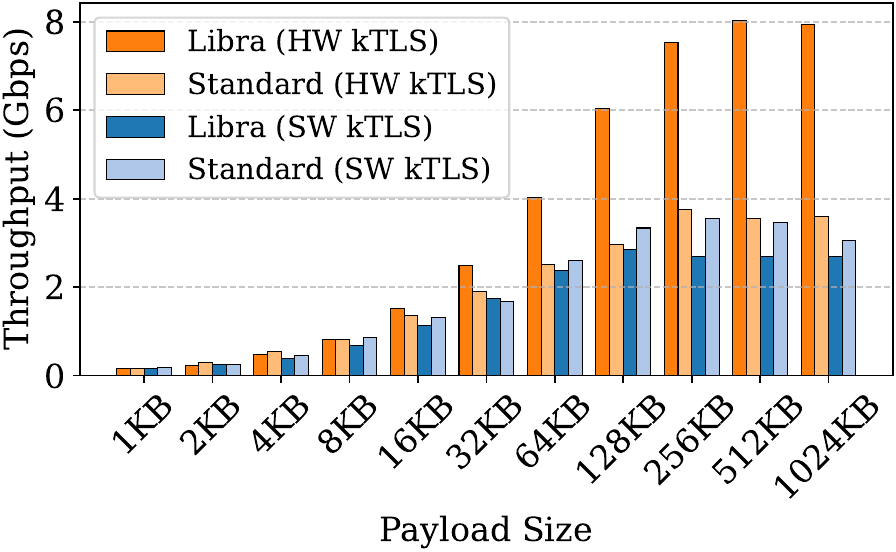}
  \caption{kTLS: Throughput}
  \label{fig:ktls_throughput}
\end{subfigure}
\hfill
\begin{subfigure}{0.32\textwidth}
  \centering
  \includegraphics[width=\linewidth]{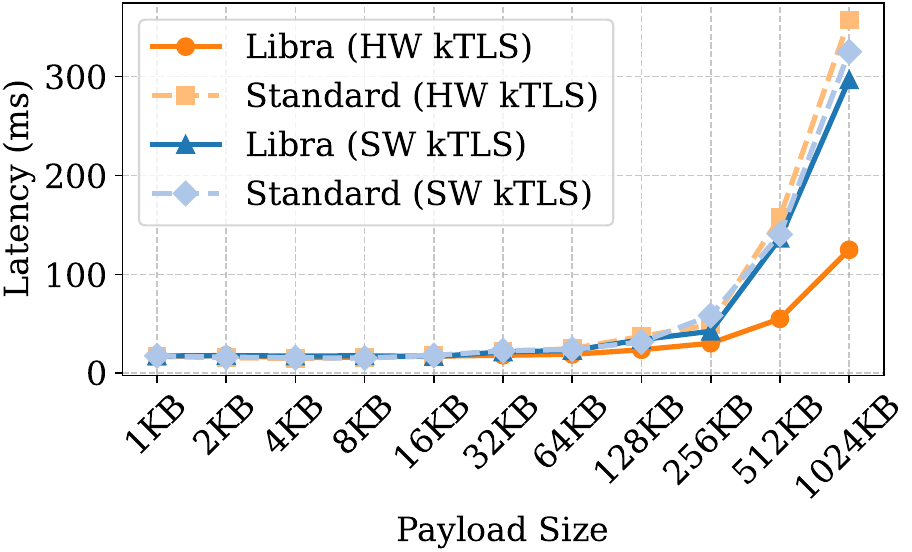}
  \caption{kTLS: P99 Latency}
  \label{fig:ktls_latency}
\end{subfigure}
\hfill
\begin{subfigure}{0.32\textwidth}  
  \centering
  \includegraphics[width=\linewidth]{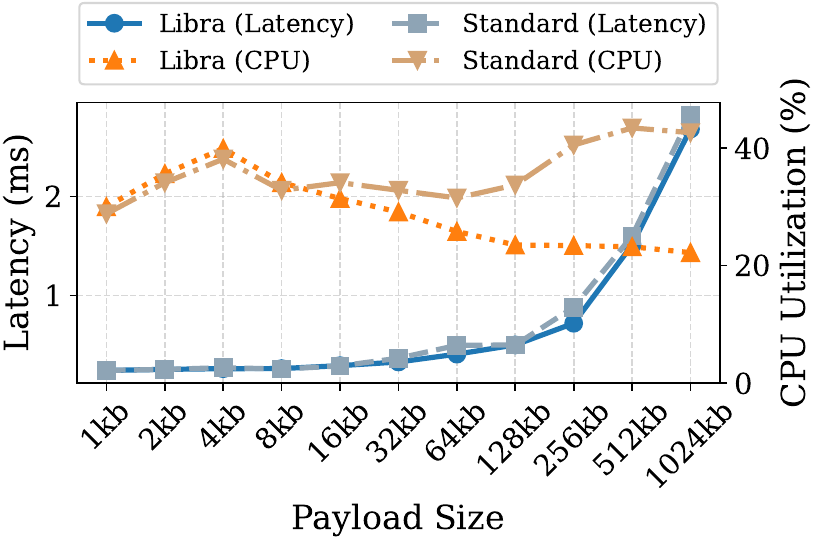}
  \caption{Single-stream performance: Latency (left axis) and CPU utilization (right axis).}
  \label{fig:single_stream_eval}
\end{subfigure}

\caption{
  Performance comparison of Libra vs. standard stack (and F-Stack/kTLS variants).}
\label{fig:all_performance_comparison}
\end{figure*}

\subsection{Application Benchmarks}
\label{application}

\subsubsection{Throughput and Tail Latency}
\label{subsubsec:tpt and lat}
To reflect realistic cloud L7 proxy workloads, we evaluate Libra using two widely deployed applications, Nginx and HAProxy, under 64 concurrent connections. This concurrency level ensures that the CPU is fully saturated. We deploy Nginx and HAProxy with Libra, as well as their counterparts, using the native Linux standard API.
To compare against kernel-bypass architectures, we introduce F-Stack, a DPDK-based user-space networking stack that provides an integrated Nginx implementation. We have carefully tuned F-Stack to ensure it operates at its optimal performance state. Since F-Stack confines both packet I/O and application logic to the same CPU core, we further introduce a variant of Libra with Nginx, denoted \texttt{Libra-Pinned}, in which both the proxy worker process and Linux receive softirqs are explicitly pinned to the same core.

Figures~\ref{fig:throughput_barchart} and~\ref{fig:latency_p99_linechart} show throughput and 99th-percentile (P99) latency across payload sizes from 1\,KB to 1024\,KB.

\textbf{Throughput.}  
For small payloads ($\leq$16 KB), the overhead of metadata parsing in Libra offsets the benefit of avoiding tiny copies, resulting in throughput comparable to or lower than the standard stack. \texttt{Libra-Pinned} performs worst at this scale due to contention from softirq and application co-location. In contrast, F-Stack achieves the highest throughput by eliminating system calls and kernel context switches, outperforming all other configurations.

As payload size increases, Libra’s advantage becomes evident. Starting at 32 KB, it surpasses the standard stack and continues to grow. At 1024 KB, Libra achieves 4.2$\times$ and 2.2$\times$ higher throughput than the baseline for Nginx and HAProxy, respectively. Meanwhile, F-Stack's performance degrades with larger payloads, and is overtaken by \texttt{Libra-Pinned} at 128 KB. The \texttt{Libra-Pinned} variant further improves performance, reaching nearly 4$\times$ the throughput of F-Stack at 1024 KB.

\textbf{Tail Latency.}  
For payloads up to 128 KB, Libra exhibits P99 latency similar to the standard stack. However, as payload size grows beyond 256 KB, proxies using the standard stack experience rising tail latency due to increased data copying. In contrast, Libra maintains stable, payload-independent processing time. At 1024 KB, it reduces P99 latency by 91\% compared to standard Nginx and by over 60\% compared to standard HAProxy.

F-Stack achieves the lowest tail latency for small payloads, maintaining values near 4 ms. At 32 KB, its P99 latency is less than 20\% of that of \texttt{Libra-Pinned}. However, its tail latency spikes with increasing payload size, while \texttt{Libra-Pinned} still sustains consistently low latency. At 1024 KB, it achieves only 13\% of F-Stack’s P99 latency.

\textbf{Why DPDK performs worse.}
This inversion stems from F-Stack’s fragmented memory model: large payloads must be split across multiple fixed-size buffers (typically $2$\,KB). Because data cannot be processed in a contiguous block, F-Stack incurs repeated per-buffer operations, preventing efficient use of CPU pipelining. This leads to significant overhead and degraded performance as payload size grows.

\subsubsection{Performance under kTLS with Nginx}
\label{subsubsec:ktls}

We evaluate Libra in a widely deployed TLS termination scenario: clients connect to the proxy over HTTPS, and the proxy forwards requests to backend servers via HTTP~\cite{prism, qdsr}. We use ECDHE-RSA-AES128-GCM-SHA256 as the cipher suite and test both software (SW) and hardware (HW) kTLS configurations. Since Nginx is currently the only mainstream L7 proxy with native kTLS support, we focus on it in this evaluation.
Figure~\ref{fig:ktls_throughput} and Figure~\ref{fig:ktls_latency} present throughput and 99th-percentile latency across payload sizes from 1\,KB to 1024\,KB under 64 concurrent connections.

\textbf{Throughput.} In software kTLS (SW) mode, Libra achieves lower throughput than the standard stack across all payload sizes. This performance degradation is due to data fragmentation inherent in our zero-copy mechanism, which reduces AES-NI efficiency~\cite{aes-ni}. We defer the detailed architectural analysis of this SW kTLS bottleneck to Appendix~\ref{sec:appendix_sw_ktls}.

In hardware kTLS (HW), Libra outperforms the standard stack starting at 16\,KB, with speedup growing to 2.15$\times$ at 1024\,KB. Although encryption is offloaded to the NIC, the standard stack remains bottlenecked by memory copying, eliminated by Libra to fully realize hardware acceleration. The gain is slightly lower than in plaintext due to Nginx’s 16\,KB TLS record limit, which requires a separate send syscall per record and offsets zero-copy benefits.

\textbf{Tail Latency.} As shown in Figure~\ref{fig:ktls_latency}, Libra reduces P99 latency under HW kTLS, with the benefit becoming pronounced for larger payloads. At 1024\,KB, it achieves a 65.0\% reduction in P99 latency compared to the standard stack.
In contrast, under SW kTLS, Libra exhibits higher P99 latency than the standard stack at large payloads. This aligns with the throughput trend and is similarly driven by the aforementioned software encryption inefficiencies. 

In particular, in the standard stack, the performance gap between software (SW) and hardware (HW) kTLS is minimal. This is because, although HW offloads encryption, its benefit is offset by costly send-side copying and more fragmented memory allocation, which increases overhead. A detailed analysis is provided in Appendix~\ref{sec:hw-sw}. Libra breaks this barrier, unlocking the true advantage of HW kTLS.

\begin{figure*}[t]
\centering
\begin{minipage}[b]{0.32\textwidth}
    \centering
    \includegraphics[width=0.95\linewidth]{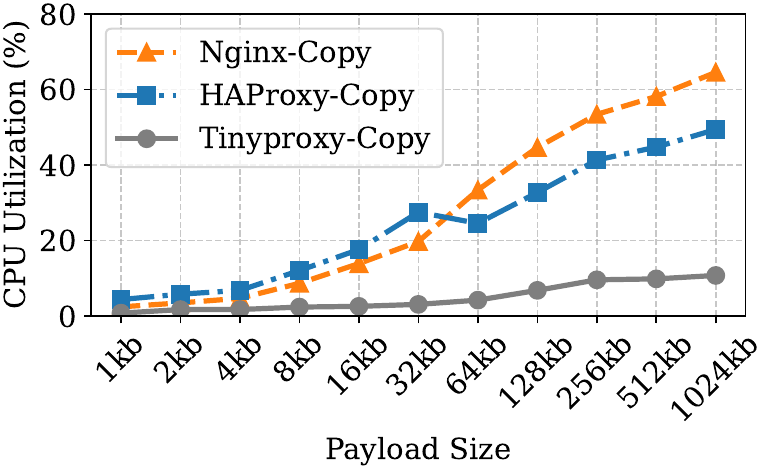}
    \captionsetup{skip=18pt, position=bottom}
    \caption{Percentage of CPU cycles spent on data copying across different proxies.}
    \label{fig:tinyproxy_overhead}
\end{minipage}%
\hfill
\begin{minipage}[b]{0.66\textwidth}
    \centering
    
    \addtocounter{figure}{1} 
    
    \begin{subfigure}[b]{0.49\linewidth}
        \includegraphics[width=0.95\linewidth]{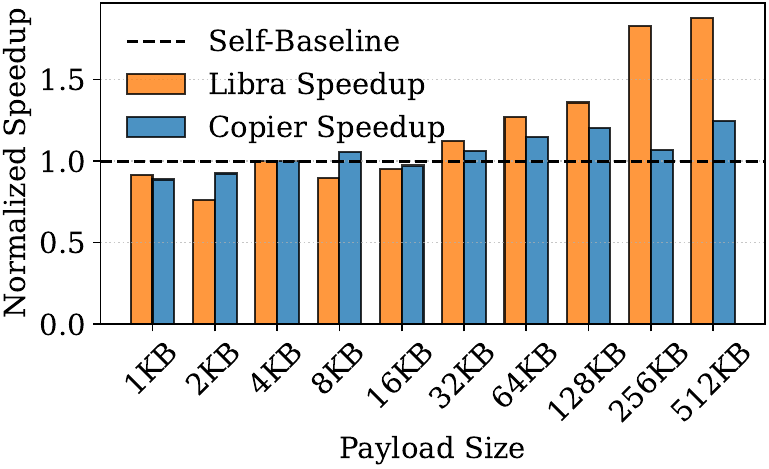}
        \caption{Normalized speedup under 8-connection.}
        \label{fig:copier_8c}
    \end{subfigure}%
    \hfill
    \begin{subfigure}[b]{0.49\linewidth}
        \includegraphics[width=0.95\linewidth]{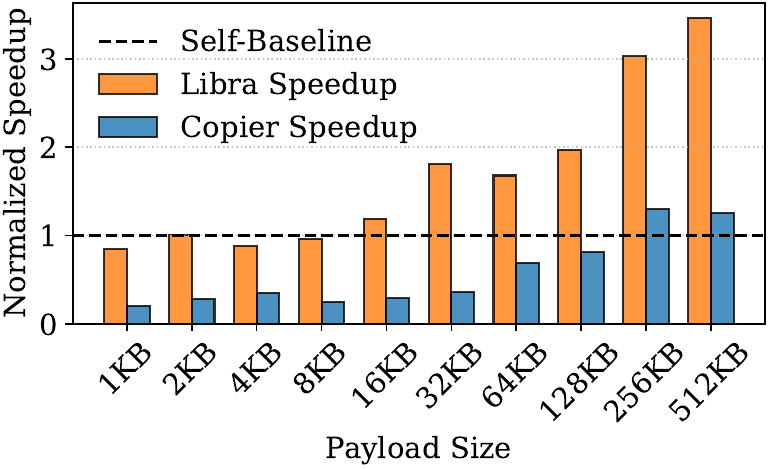}
        \caption{Normalized speedup under 64-connection.}
        \label{fig:copier_64c}
    \end{subfigure}
    
    \addtocounter{figure}{-1} 
    
    \caption{Performance comparison of Libra vs. Copier. %
    \normalfont\small\emph{Speedup normalized to each system’s baseline to account for kernel differences (Libra: Linux 6.11; Copier: 5.15).}%
    }
    \label{fig:speedup_comparison}
\end{minipage}
\end{figure*}

\subsubsection{Single-Stream Processing Efficiency} 
\label{subsubsec:single-stream}

To evaluate Libra's performance under a single-stream scenario, we conduct a single-stream benchmark using Nginx. Figure~\ref{fig:single_stream_eval} shows the average end-to-end latency (from request issuance to full response receipt) and CPU utilization of Nginx using Libra versus using the native Linux standard API across payload sizes ranging from 1\,KB to 1024\,KB.

\textbf{Latency.} As shown in Figure~\ref{fig:single_stream_eval}, Libra achieves nearly identical end-to-end latency compared to the standard stack. This is because the total latency is dominated by components outside the kernel data path—such as network propagation and server-side processing—leaving little room for the reduction in kernel copying to manifest in the end-to-end measurement.

\textbf{CPU Efficiency.} For small payloads ($\leq$8\,KB), Libra’s CPU utilization is slightly higher than that of the native Linux standard API, because the overhead of metadata parsing is fixed, while small-scale memory copying incurs almost no additional cost. However, once the payload size exceeds 16\,KB, the advantage of zero-copy begins to emerge. The standard baseline’s CPU overhead increases with payload size, while Libra’s CPU utilization gradually decreases and stabilizes. At a 1024\,KB payload, Libra reduces CPU utilization from 42.6\% to 22.2\%, a reduction of 47.9\%.

\subsection{State-of-the-Art Comparison}
\label{subsubsec:copier_comparison}

To ensure a fair and controlled comparison with Copier, we implement a custom high-performance, epoll-driven, multi-threaded L7 proxy. This is necessary because Copier’s original evaluation uses TinyProxy~\cite{tinyproxy2026}, whose bottleneck stems from inefficient user-space logic rather than data copying. We compare the proportion of CPU overhead due to copying in TinyProxy versus production-grade proxies. As shown in Figure~\ref{fig:tinyproxy_overhead}, the copy-related overhead in TinyProxy remains below 15\% across all payload sizes—significantly lower than that of Nginx and HAProxy—thereby masking the benefits of zero-copy techniques.

While Nginx and HAProxy are highly optimized, their I/O pipelines are deeply entangled with application logic, and adapting them to Copier’s APIs is complex—making it difficult to isolate zero-copy effects without confounding artifacts. Our minimal proxy cleanly separates network I/O from processing logic, enabling a precise and direct comparison with Copier.

To account for kernel version differences (Libra on Linux 6.11.0, Copier on 5.15.131), we perform baseline measurements using the standard stack on both kernels and compute normalized speedup relative to each system’s standard stack baseline to ensure a fair comparison independent of kernel-level performance variations.
Figure~\ref{fig:copier_8c} and~\ref{fig:copier_64c} show the normalized throughput speedup across payload sizes from 1\,KB to 512\,KB under 8 and 64 concurrent clients.

\begin{figure*}[t]
  \centering
  \begin{subfigure}[b]{0.3\textwidth}
    \centering
    \includegraphics[width=\linewidth]{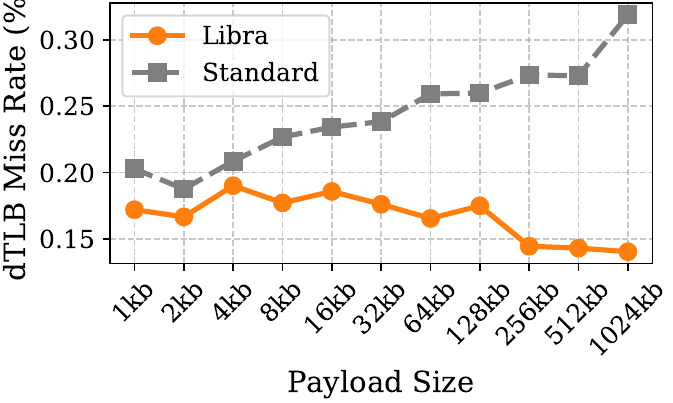}
    \caption{dTLB miss rate.}
    \label{fig:tlb_miss_rate}
  \end{subfigure}
  \hfill
  \begin{subfigure}[b]{0.36\textwidth}
    \centering
    \includegraphics[width=\linewidth]{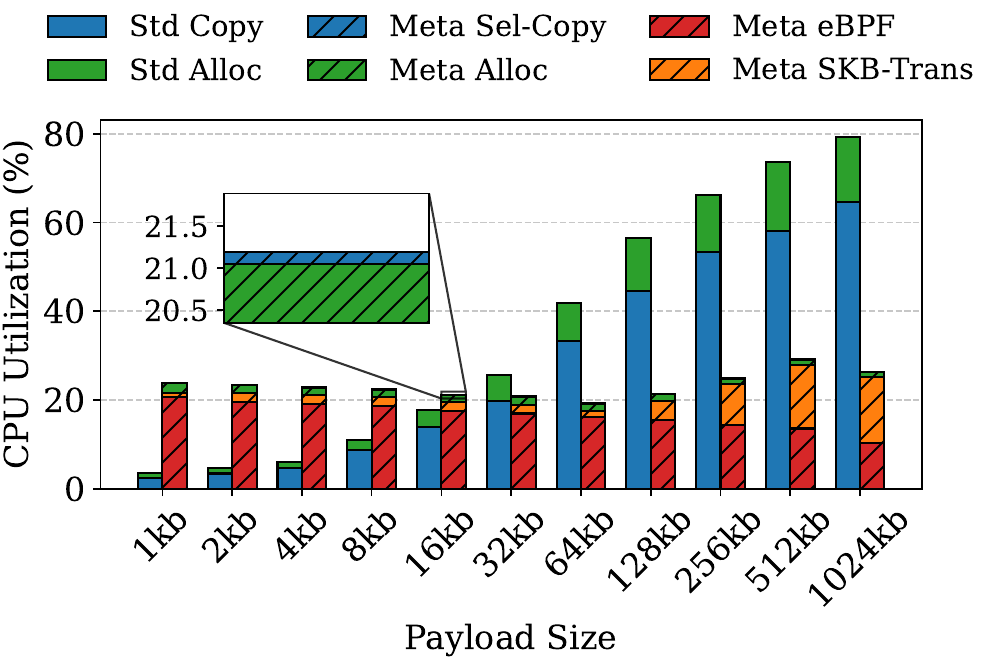}
    \caption{CPU utilization.}
    \label{fig:overhead_cpu}
  \end{subfigure}
  \hfill
  \begin{subfigure}[b]{0.31\textwidth}
    \centering
    \includegraphics[width=\linewidth]{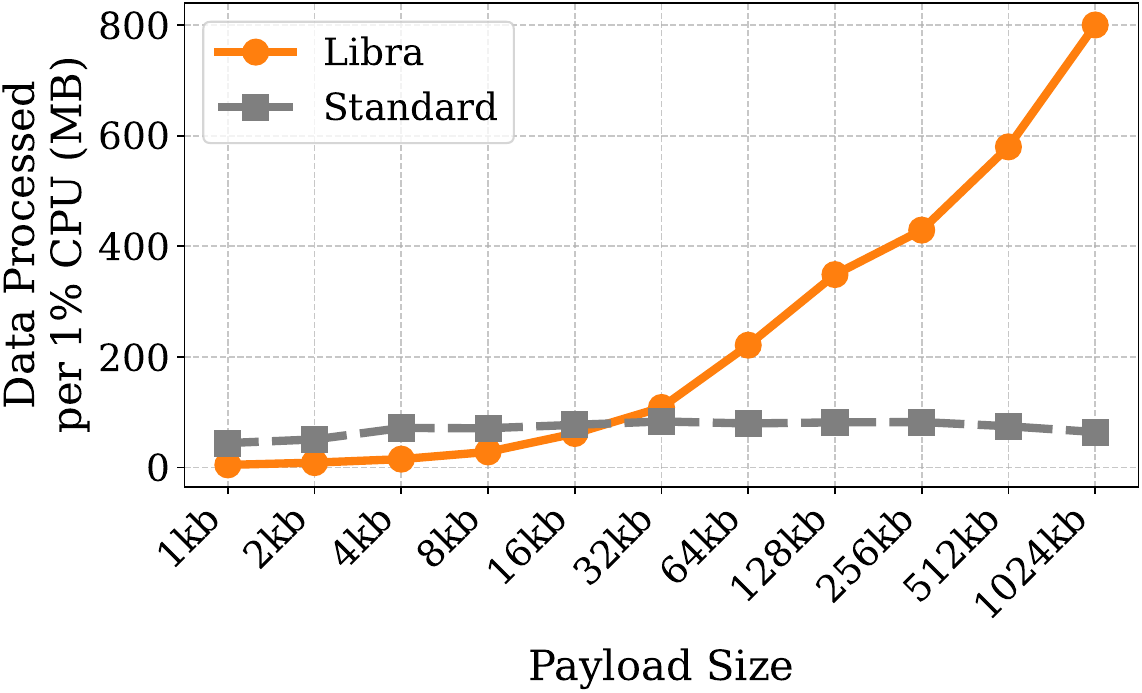}
    \caption{Data processed per 1\% CPU utilization.}
    \label{fig:overhead_efficiency}
  \end{subfigure}
  \caption{Comparative analysis between Libra and the standard stack across different performance metrics.}
  \label{fig:comparative_analysis}
\end{figure*}


\textbf{Low concurrency (8 connections).} 
Under 8 concurrent connections, both Libra and Copier show limited speedup for small payloads ($\leq$32\,KB), where baseline copy overhead is negligible. Gains emerge at 32\,KB: Libra achieves 1.11$\times$ speedup and scales to 1.97$\times$ at 512\,KB. Copier improves more modestly (e.g., 1.25$\times$ at 512\,KB) by batching \texttt{recv}/\texttt{send} into a single asynchronous copy. However, as discussed in \S\ref{sec:copies_in_modern_l7_proxies}, the dominant cost in modern L7 proxies stems from receive-side copying over fragmented skbs. Since Copier still performs a full memory copy across these fragmented buffers, it remains constrained by the bottleneck.

\textbf{High concurrency (64 connections).} 
At 64 concurrent connections, the architectural divergence becomes stark. Libra fully saturates the CPU and maintains strong speedup across all payload sizes from 16\,KB onward, reaching 3.5$\times$ at 512\,KB. In contrast, Copier’s speedup collapses, especially for small payloads (e.g., only 0.2$\times$ at 1\,KB), due to severe lock contention. Copier relies on a single kernel thread to process all copy requests, requiring synchronization on every submission. Under high concurrency, frequent \texttt{recv}/\texttt{send} calls—especially for payloads $\leq$ 128\,KB, as clients quickly consume responses and issue new requests—cause intense lock contention, severely degrading performance.

\subsection{Micro-Architectural Efficiency}
\label{sec:micro_efficiency}
To evaluate the micro-architectural impact of Libra, we run Nginx as an L7 proxy with 64 concurrent clients over 5 seconds, comparing Libra against the standard stack. We first measure dTLB miss rates, then exclude shared costs (e.g., TCP processing) to isolate the CPU overhead specific to Libra and the standard stack, respectively. Based on their isolated totals, we quantify how much data each processes per 1\% of its own stack-specific CPU overhead.

\textbf{dTLB miss rate.} 
As shown in Figure~\ref{fig:tlb_miss_rate}, Libra reduces the dTLB miss rate across all payload sizes. This stems from its elimination of redundant memory copies: by anchoring payload pages in the kernel and directly transferring skb ownership between sockets, Libra preserves virtual-to-physical address mappings throughout the packet’s lifetime. In contrast, the standard stack performs separate copies on the receive and send paths, each triggering new page walks and increasing TLB pressure.

\textbf{CPU overhead analysis.}
Figure~\ref{fig:overhead_cpu} quantifies the CPU overhead breakdown, categorized as follows:

\texttt{Std Copy}: full data copying in \texttt{recv}/\texttt{send};

\texttt{Std Alloc}: memory allocation for full payloads;

\texttt{Meta Sel-Copy}: selective copying of new L7 metadata;

\texttt{Meta Alloc}: memory allocation for new L7 metadata;

\texttt{Meta eBPF}: eBPF program;

\texttt{Meta SKB-Trans}: zero-copy reuse of anchored payloads.

Libra exhibits payload-independent overhead, remaining relatively stable between 19\% and 30\%. This overhead primarily stems from the eBPF program and the reuse of anchored payloads; the cost of selective copying is negligible, and memory allocation overhead is similarly small. In contrast, the standard stack incurs copying and memory allocation overhead that scales linearly with payload size—rising from 4\% at 1\,KB to 80\% at 1024\,KB—a cost effectively eliminated by Libra.

\textbf{CPU efficiency.}
Figure~\ref{fig:overhead_efficiency} compares the efficiency of Libra and the standard stack: at small payloads (e.g., 1\,KB), Libra is less efficient (4.6\,MB vs.\ 44\,MB per 1\% CPU), due to the fixed cost of metadata parsing dominating the total overhead. 
As payload size increases, Libra’s efficiency improves by avoiding per-byte copying, while the standard stack remains stable due to linearly scaling costs of memory allocation and data copying.
Consequently, at 1024\,KB payloads, Libra processes up to 800\,MB of data per 1\% CPU utilization, whereas the standard stack handles only 79\,MB, a more than 10$\times$ improvement.

\section{Discussion}
\label{sec:discussion}



\textbf{Extending to Multiplexed Protocols (HTTP/2)}  
While our prototype targets serial protocols like HTTP/1.1, HTTP/2 multiplexes multiple streams over a single TCP connection, making selective bypass challenging: a monolithic receive queue incurs complex parsing and head-of-line blocking. To address this, we propose extending Libra with an \textit{eBPF-driven per-stream queue}. The Virtual Payload Identifier (VPI) becomes stream-scoped, derived from the socket pointer and the HTTP/2 stream id. When the control plane detects an interleaved frame, it maps the VPI to a dedicated stream queue, and the data plane relinks the corresponding payload segments from the main connection queue without deallocation or deep copying. This preserves zero-copy anchoring while enabling efficient in-kernel stream demultiplexing.

\textbf{Generalizing to Kernel-Bypass Architectures}  
While Libra is currently implemented in the Linux kernel, its core paradigm—decoupling a lightweight, protocol-aware control plane from the data movement path—is not inherently tied to the kernel. This separation could potentially be applied to kernel-bypass frameworks with user-space network stacks~\cite{mtcp, fstack, seastar}, where opportunities for zero-copy payload handling and metadata offloading may exist.

\textbf{Limitations of Software Cryptography}
Although Libra currently underperforms the standard stack under software (SW) kTLS due to the inefficiency of AES-NI when encrypting highly fragmented \texttt{sk\_buff}s, future optimizations to AES-NI or kernel cryptographic routines specifically designed for fragmented scatter-gather lists could significantly mitigate this overhead. Moreover, the industry is gradually moving toward hardware-accelerated cryptography. As hardware-offloaded kTLS becomes increasingly ubiquitous in cloud infrastructures, Libra will be well-positioned to fully leverage this hardware trend, completely bypassing software encryption bottlenecks and thereby delivering substantial zero-copy performance gains.
\section{Conclusion}
\label{sec:conclusion}

We present Libra, an OS-level selective copy framework that achieves near zero-copy performance for Layer-7 socket I/O. We prove that true zero-copy is fundamentally incompatible with full POSIX compatibility in mainstream operating systems. Libra leverages eBPF to deliver only the control-plane metadata to the user space while anchoring the opaque bulk payload in the kernel, enabling the proxy to reuse the anchored payload on the egress path while preserving full application transparency. Our results show that this approach yields substantial performance gains for widely-used, unmodified proxy applications.

\bibliographystyle{ACM-Reference-Format}
\bibliography{main}

\clearpage
\appendix
\section{IMPLEMENTATION DETAILS}
\label{sec:others}

\subsection{Receive Socket Memory Management}
\label{sec:recv_sock}
Payload anchoring retains uncopied data in the kernel, increasing socket memory usage and causing the OS to automatically shrink the TCP receive window. Since the theoretical per-socket buffer capacity provided by the OS is relatively generous, Libra dynamically enlarges the socket’s receive buffer to maintain the advertised window size and avoid network throttling. 

To avoid OOM conditions from sustained anchoring of large unforwarded payloads, Libra enforces a hard limit on anchored payload size via a user-configurable eBPF threshold (default: 3~MB). This provides ample zero-copy capacity for the vast majority of service-to-service messages while bounding kernel memory consumption per connection. Once the threshold is reached, the eBPF program truncates the expected message length in its state map. The safely anchored prefix remains zero-copy, while subsequent data falls back to standard deep copying—safely shifting excess memory pressure to user space.

\subsection{Deadlock-Free Socket Transfers}
\label{sec:deadlock}
Zero-copy forwarding of \texttt{sk\_buff}s from one socket to another requires manipulating pointers across two socket structures. In bidirectional proxy traffic, holding locks on both sockets simultaneously to perform this transfer may lead to an AB-BA deadlock.

To fully decouple the lock acquisition order, we introduce an Intermediate Staging Queue. The kernel first acquires the lock on the receive socket, extracts the requested amount of anchored payload, and enqueues it into this intermediate queue before explicitly releasing the receive socket lock. It then acquires the lock on the send socket, performs a selective copy of the metadata, and drains the intermediate queue to transfer payload ownership into the send socket’s write queue. This phased approach guarantees deadlock-free cross-socket transfers.

\subsection{Send Socket Memory Accounting}
\label{sec:send_sock}
Our zero-copy forwarding path uses an intermediate staging queue to hold payloads pending transfer from the receive socket to the send socket, with its size strictly bounded by the requested length. However, during actual transmission, the send socket’s write queue may already contain data, and its limited send buffer may only accommodate a partial transfer—preventing the entire staged payload from being enqueued at once. Re-enqueuing the untransferred portion back into the receive socket’s queue is not viable, as it would exacerbate lock contention on the receive socket, which we explicitly aim to minimize.

The kernel imposes a per-socket send buffer limit to prevent any single connection from exhausting system memory\footnote{%
This mechanism only governs admission into the socket’s send queue based on local memory accounting. TCP congestion control and receiver-advertised flow control are enforced later in the protocol stack during actual packet transmission, and are unaffected by our temporary adjustment of the send buffer limit.%
}. This limit relies on strict memory accounting that assumes all data enqueued to a socket is owned by that socket—a valid assumption under the standard copy-based forwarding model, where data is explicitly allocated into the send buffer upon transmission.

In Libra, however, anchored payloads physically reside in receive-side buffers (e.g., NIC page pools) and are not allocated into the send socket’s memory. Nevertheless, to maintain compatibility with the kernel’s existing accounting logic, we explicitly count the staged payload against the send socket’s memory budget. Without this accounting entry, the subsequent decrement performed by the lower protocol stack upon transmission completion would cause an underflow.

Consequently, even though zero-copy payloads consume no additional kernel memory, they can still be blocked if the send socket’s write queue is logically full under the current buffer limit. To enable atomic forwarding of the entire requested payload, Libra, in \texttt{FAST\_PATH}, temporarily raises the send socket’s buffer limit by exactly the size of the staged payload. This adjustment is safe: since no new memory is allocated, there is no risk of actual memory overcommitment. The payload is then submitted directly to the network stack.

\subsection{Safe Socket Teardown}
\label{sec:close}
Zero-copy forwarding introduces a critical safety challenge: if the proxy or client invokes \texttt{sys\_close} before the \texttt{sk\_buff} ownership transfer completes, the underlying \texttt{recv\_sock} may be released while still referenced by in-flight zero-copy data, leading to a dangling pointer.

To address this, we implement a reference-count-based anchoring mechanism with deferred cleanup. The ingress side makes the anchoring decision: when the \texttt{RX-Prog} injects a virtual payload identifier (VPI) into the user-space buffer—indicating that the payload will be forwarded via zero-copy—the kernel calls \texttt{sock\_hold} to increment the reference count of the \texttt{recv\_sock}. This logical anchor is recorded in an eBPF map, ensuring the socket remains valid in kernel memory even if the application closes it prematurely.

Upon detecting a \texttt{sys\_close} event, the system checks the eBPF map for an active anchor. If present, the connection enters a \textit{Deferred Teardown} state and starts a configurable grace period (default: 5 seconds). Once the timer expires, the system resets the connection’s eBPF state machine to \texttt{DEFAULT}, removes the associated VPI mapping and anchor entry from the global eBPF map, and finally releases the socket through the standard deallocation path.

\subsection{Mitigating Context Switches via Fragment Limit Tuning}
\label{bigtcp}
Libra suffers from a hidden asymmetry between RX and TX \texttt{sk\_buff} construction. On the receive side, each incoming MSS segment (e.g., 1448 bytes) is stored in a separate page fragment due to driver and memory pool constraints. With the default \texttt{MAX\_SKB\_FRAGS} of 17, GRO can aggregate at most 17 such segments—roughly 25–26KB—into one \texttt{sk\_buff}. 

In contrast, the transmit path leverages high-order page allocations: a single fragment can span tens of kilobytes (e.g., 32KB), allowing one \texttt{sk\_buff} to carry up to the \texttt{MAX\_GSO\_SIZE} limit of 64KB. For a 1448-byte MSS, this means a single TX \texttt{sk\_buff} typically contains 45 MSS segments ($45 \times 1448 = 65{,}160$ bytes), often packed into one or two large fragments.

When Libra forwards RX \texttt{sk\_buff}s directly, it inherits their small size. Consequently, delivering the same 64KB of data requires more than twice as many \texttt{sk\_buff}s compared to native TX. Since the NIC processes transmissions per \texttt{sk\_buff}, a fast downstream client issues proportionally more \texttt{recv} syscalls, causing excessive context switches—an overhead that negates zero-copy CPU savings.

To bridge this gap, we increase \texttt{MAX\_SKB\_FRAGS} to exactly 45 while keeping \texttt{MAX\_GSO\_SIZE} at 64KB, inspired by BIG TCP~\cite{bigtcp}. This allows GRO to aggregate 45 MSS segments into a single RX \texttt{sk\_buff}, matching the payload size used by TX. The forwarded \texttt{sk\_buff} then passes egress GSO checks seamlessly and avoids unnecessary downstream context switches.

\section{ANALYSIS OF KTLS PERFORMANCE}
\label{sec:appendix_ktls}

\subsection{Why Libra Slows SW kTLS}
\label{sec:appendix_sw_ktls}

As previously observed, Libra exhibits lower throughput than the standard stack under software kTLS (SW) mode. To understand this behavior, we compare the CPU overhead of key operations while excluding shared components. Specifically, we analyze the extra overhead (\texttt{Meta Extra}) and encryption cost (\texttt{Meta AES-NI}) under Libra against the copy overhead (\texttt{Std Copy}) and the encryption cost (\texttt{Std AES-NI}) in the standard stack. As shown in Figure~\ref{fig:aes}, the total overhead in Libra is higher than the standard stack.

\begin{figure}[htbp]
   \centering
   \includegraphics[width=0.95\linewidth]{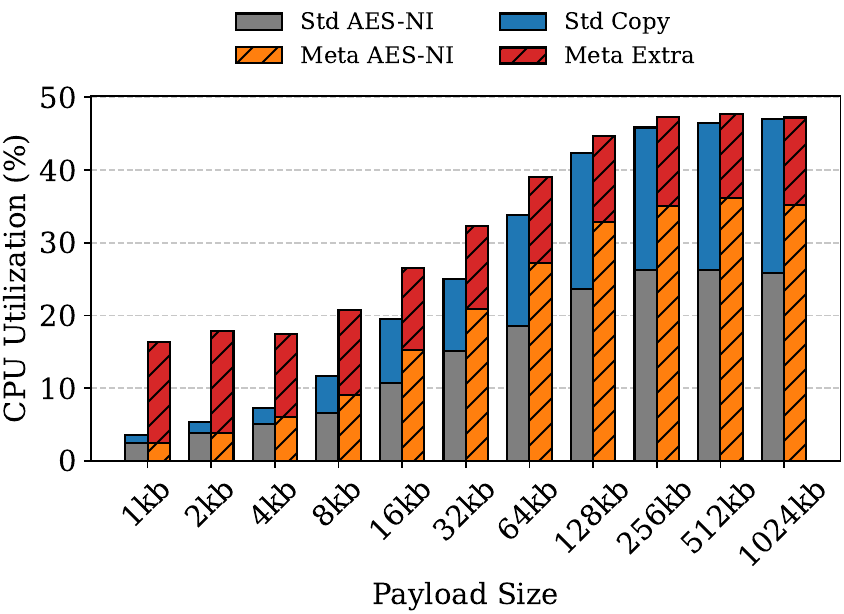}
   \caption{Analysis of CPU overhead components under software kTLS.}
   \label{fig:aes}
\end{figure}

\begin{figure}[htbp]
   \centering
   \includegraphics[width=\linewidth]{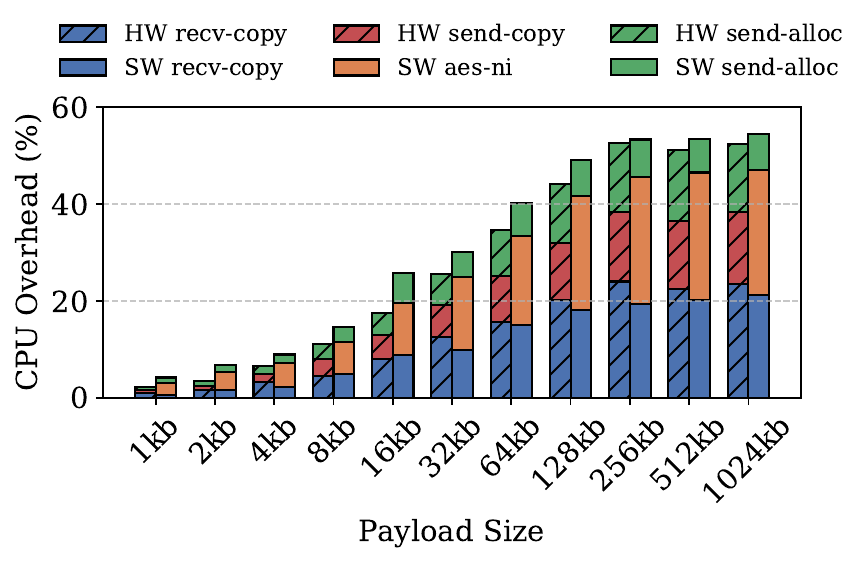}
   \caption{Analysis of CPU overhead components: kTLS HW vs. SW in the standard stack.}
   \label{fig:hw_sw}
\end{figure}

The reason is that, in the current SW kTLS implementation, CPU-based encryption effectively functions as an ``encrypt-and-copy'' operation at the hardware level. When using the AES-NI instruction set~\cite{aes-ni}, the CPU must load the plaintext into its registers, execute the encryption rounds, and write the resulting ciphertext back to a separate memory buffer. This mandatory output phase is structurally and cost-wise equivalent to a memory copy. Under existing hardware execution pipelines, this behavior cannot be bypassed by any software-level zero-copy mechanism. Consequently, Libra primarily eliminates the redundant receive-side copy (\texttt{Std Copy}), while the encryption-related copy on the send side remains unavoidable.

\subsection{Why HW kTLS Gains Little}
\label{sec:hw-sw}
Furthermore, eliminating the receive-side copy forces the AES-NI engine to process plaintext directly from the fragmented \texttt{sk\_buff} structures in the receive socket. In the standard Linux network stack, the payload is copied into a contiguous user-space buffer, enabling the AES-NI pipeline to benefit from high spatial locality and efficient hardware prefetching. In contrast, under Libra, the payload is physically scattered across non-contiguous page fragments. This fragmentation disrupts current hardware prefetching logic and induces frequent CPU stalls. We observe that this lack of contiguity significantly increases encryption cost, causing Libra’s AES-NI overhead (\texttt{Meta AES-NI}) to exceed that of the standard stack (\texttt{Std AES-NI}).

Ultimately, Libra’s performance in SW kTLS represents a direct cost trade-off: while it successfully saves CPU cycles by bypassing the receive-side memory copy, it introduces two additional penalties—the fixed overhead of Libra (\texttt{Meta Extra}) and the AES-NI fragmentation penalty.

To understand why hardware kTLS (HW) provides little performance improvement over software kTLS (SW) in the standard stack, we compare the CPU overhead that is unique to each path, excluding shared costs such as TCP processing or application logic.

In the HW path, encryption is offloaded to the NIC, so the CPU no longer executes AES-NI. It performs the following:
The HW path performs the following:
\begin{itemize}
    \item \texttt{recv-copy}: copies received plaintext from kernel buffers to user space;
    \item \texttt{send-copy}: copies outgoing data into kernel page fragments to construct TLS records;
    \item \texttt{send-alloc}: allocates memory for plaintext and TLS metadata.
\end{itemize}

In contrast, the SW path performs encryption on the CPU using AES-NI (\texttt{aes-ni}) and avoids \texttt{send-copy} entirely: it encrypts data in-place within pre-allocated kernel buffers. Its overhead, therefore, consists of:
\begin{itemize}
    \item \texttt{recv-copy}: same as HW;
    \item \texttt{aes-ni}: for encryption;
    \item \texttt{sw-send-alloc}: memory allocation for ciphertext.
\end{itemize}

Figure~\ref{fig:hw_sw} shows the total CPU overhead of these path-specific components. Across all payload sizes, the overhead gap between HW and SW remains small. The narrow difference explains why hardware offload provides little benefit in the standard stack.

This behavior stems from two key issues in the HW path, which we identified through inspection of the Linux kernel source code:

\textbf{Fragmented memory allocation due to aggressive packing:}  
HW constructs TLS records directly in page fragments and uses an aggressive packing policy that fills the current fragment to its last available byte before allocating a new page. Because the socket’s page fragment is reused across records and often starts with partial occupancy, this strategy leads to frequent, small allocation requests—even for large payloads. These repeated calls prevent the allocator from staying in its fast path, increasing \texttt{send-alloc} overhead. In contrast, SW allocates large, contiguous buffers upfront for ciphertext, resulting in far fewer allocation events.

\textbf{Uncached send-side copying in HW:}  
The \texttt{send-copy} operation in the HW path uses cache-bypassing memory writes by default\kr{cite}, which bypass the CPU cache hierarchy and incur higher per-byte CPU overhead compared to the cached copies used in plaintext transmission. Consequently, even when the total amount of copied data is the same, the send-side copy cost in HW is substantially higher.

\end{document}